\documentclass[aps,10pt,reprint,groupedaddress,superscriptaddress,onecolumn,longbibliography]{revtex4-2}
\usepackage[usenames,dvipsnames]{xcolor}
\usepackage{amssymb}
\usepackage{graphicx}
\usepackage{amsmath}
\usepackage{dsfont}
\usepackage[bookmarks=true,colorlinks,citecolor=blue,urlcolor=blue]{hyperref}
\usepackage{amsthm}
\usepackage[capitalize,nameinlink]{cleveref}

\usepackage{braket}
\usepackage{babel}
\usepackage{blindtext}
\usepackage{physics}
\usepackage{float}
\usepackage{bm}
\usepackage{apptools}
\usepackage{caption}
\usepackage{subcaption}

\newcommand{\TW}{W_\mathrm{tree}}

\newtheorem{theorem}{Theorem}
\newtheorem{lemma}{Lemma}
\newtheorem{corollary}{Corollary}
\theoremstyle{definition}
\newtheorem{definition}{Definition}

\AtAppendix{\counterwithin{theorem}{section}}
\AtAppendix{\counterwithin{lemma}{section}}
\AtAppendix{\counterwithin{corollary}{section}}

\begin{document}

\author{Olli Hirviniemi}
\email{olli.hirviniemi@iqm.tech}
\affiliation{IQM, Keilaranta 19 D, 02150 Espoo, Finland}
\author{Afrad Basheer}
\email{afrad.basheer@iqm.tech}
\affiliation{IQM, Georg-Brauchle-Ring 23-25, 80992 München, Germany}
\author{Thomas Cope}
\email{thomas.cope@iqm.tech}
\affiliation{IQM, Georg-Brauchle-Ring 23-25, 80992 München, Germany}

\title{Random Quantum Circuits as Seeds for Continuous Generative Models}



\begin{abstract}
We introduce a random circuit family and show they are robust against current classical simulation techniques, specifically tensor network contraction and Pauli propagation. We also show that local variables do not concentrate, ensuring enough variance to be able to produce a diverse set of data points. We therefore argue that using these circuits as a ``random seed" for a larger classical generative model is a way to make large-scale quantum-classical hybrid models amenable towards NISQ devices. 
\end{abstract}

\maketitle

\section{Introduction}
\subsection{Background and Motivation}

Ever since the publication of Shor's seismic algorithm \cite{Shor_1997}, there has been an extensive research effort to understand the additional computational power that quantum systems provide, and build systems to enable this power. Despite their promise, however, the engineering challenges which must be overcome to build quantum computers capable of running large-scale, fault-tolerant algorithms remain significant, leading researchers to turn to other algorithms in order to test the power of our current noisy intermediate scale quantum (NISQ) devices \cite{Preskill_2018}. Of these, one of the most intensely studied areas has been that of \emph{variational algorithms} \cite{Cerezo_VQAReview_2021,Stechly_VQA_2024}; in which an ansatz circuit is chosen, and then its parameters variationally optimised to minimise/maximise a desired cost function. Typically these cost functions take the form of linear operators, $f(\theta)=\mathrm{Tr}[\rho U^\dagger(\theta)O U(\theta)]$. One of the areas these models have been applied to is the field of machine learning; in such models there is an additional \emph{encoding} of the data $\mathbf{x}$ into the state $\rho(\mathbf{x})$, so that the cost function takes the form $\mathbb{E}_{\mathbf{x}\sim \mathcal{D}} \mathrm{Tr}[\rho(\mathbf{x}) U^\dagger(\theta)O U(\theta)]$. This is the general form for such models \cite{Jerbi_Lin_2023}.\\

One issue with variational methods is that, in the absence of prior knowledge regarding how to initialise the variational parameters $\theta$, one typically  cannot do better than picking a starting point uniformly. It was quickly realised that, for expressive ansätze, this leads to a phenomenon known as \emph{barren plateaus} \cite{mcclean2018barren,Qi2023,Arrasmith_2022,larocca2024}. 
This issue is characterised by the gradient becoming exponentially small over almost all parameter instances; as $O(1/\varepsilon^2)$ shots are required to estimate an observable up to precision $\varepsilon$, this means that even a single update becomes exponentially demanding.\\

Consequently, several ansätze were proposed which provably do not suffer from barren plateaus, such as shallow hardware efficient circuits \cite{cerezo2021cost,Khatri2019,Zhao2021,Liu2022,Basheer2022,Letcher2023,Suzuki2023,Leone_HEA_2024}, quantum convolutional neural networks \cite{Pesah2021} and small-angle initialisations \cite{Park2023,Wang2023,Zhang_Gauss_2022}. There was a further blow to variational algorithms in \cite{cerezo2024} however, when it was pointed out that most barren-plateau free ansätze can be optimised purely classically. These fall into two further categories; either $f(\theta)=\mathrm{Tr}[\rho U^\dagger(\theta)O U(\theta)]$ can be estimated purely classically; meaning the ansatz is effectively a classical algorithm; or $f(\theta)$ requires a classical description $\rho_{\mathrm{classical}}$ of the state $\rho$ (e.g. classical shadows \cite{HuangShadow2020}) which is only attainable through the use of a quantum computer; this still offers the possibility of a quantum advantage.\\

In the context of quantum machine learning, the transform $\mathbf{x}\rightarrow \rho(\mathbf{x})_{\mathrm{classical}}$ can be thought of as a quantum-enabled feature transform; and the final performance of the algorithm depends heavily on not just the choice of encoding $\mathbf{x}\rightarrow \rho(\mathbf{x})$, but also the distribution of data $\mathbf{x}\sim \mathcal{D}$. These two factors combined can lead to \emph{both} a barren plateau due to the distribution $\mathcal{D}$, but also a fully classical algorithm. It is therefore difficult to talk about the power of a variational algorithm for quantum machine learning without explicit reference to the dataset in question. \\

In this paper, we avoid this issue by considering a different machine learning paradigm; that of \emph{generative} learning. In generative machine learning, the goal is not to simply learn a function of the data $f(\mathbf{x})$, but to produce new data samples by learning to sample from a distribution close to the true data distribution, $\mathbf{x}\sim \mathcal{D}$.

Here, we focus on the ``minimal" way a quantum computer could contribute to a generative model. Inspired by the literature on variational quantum models for supervised learning, we argue this should be a feature transform $\gamma \rightarrow \rho(\gamma)$, where $\gamma\sim \Gamma$ is classical input randomness. Since we envision this mapping being used in a quantum-classical generative model, we must extract some classical information from $\rho(\gamma)$ (which can then be passed into e.g. neural network layers). We consider here three ways to do this: a vector of expectation values $\langle \Vec{O} \rangle_{\rho(\gamma)}$, a set of classical representations of local subsystems $\Vec{\hat{\rho}}_{i}(\gamma)$, obtained by e.g. tomography or classical shadows (\cite{HuangShadow2020,aaronson2018shadow,PRXQuantum.2.030348,Basheer2022}), or by applying a trainable quantum circuit to $\rho(\gamma)$ and then taking expectation values - effectively giving us $\langle \Vec{O}(\theta) \rangle_{\rho(\gamma)}$. It should be noted by the arguments in \cite{cerezo2024}, we do not expect a ``further" quantum advantage from this trainable layer; it should be considered more akin to a ``quantum inspired approach".\\

In order for our quantum feature map to contribute meaningfully, there are two necessary properties it must have. It must be difficult to simulate classically, and it must not suffer from a ``barren plateau". Since there are no trainable parameters in the feature map; it is to reasonable to ask how a barren plateau could appear. The answer comes in the classical information extracted from $\rho(\gamma)$; if the map is too expressive, it is possible for the information $\Vec{\hat{\rho}}_{i}(\gamma)$ to concentrate and become (when using only polynomial precision) effectively indistinguishable for all $\gamma$. In the context of generative learning, this would lead to a problem known as \emph{mode collapse}, in which the model could, at best, only produce a single realistic data instance (imagine it producing the same image of a cat repeatedly) since the seed randomness is not carried to the final output. From here onwards, we will refer to this problem as mode collapse, and reserve the term barren plateau for when considering the trainable addition to the model.\\

The technical part of this paper is therefore dedicated to showing that our choice of feature map $\gamma \rightarrow \rho(\gamma)$ does not suffer from this issue, and is robust against two leading simulation methods, namely tensor networks \cite{PhysRevLett.91.147902,cirac2009tensor,Markov_Shi_2008} and Pauli propagation \cite{PhysRevA.99.062337,shao2024simulating,angrisani2024classically}. We also acknowledge that various aspects of this argument appear elsewhere in the literature. \cite{gao_enhancing_2022} considers enhancing (classical) generative Bayesian networks using entangled states which violate Bell inequalities; the idea of using classical neural networks to compress high dimensional data into a low dimensional space for a quantum model to learn appears in e.g. \cite{perdomo-ortiz_opportunities_2018,changLatentStylebasedQuantum2024}. Using a quantum generator's expectation values for continuous data appears in \cite{romero_variational_2021,bravo-prieto_style-based_2022,changLatentStylebasedQuantum2024,Barthe2025EVS}. The
alternative approach of discretising the data was considered in \cite{zoufal_quantum_2019}, and further explored in various forms in \cite{borras_impact_2023}. \cite{bravo-prieto_style-based_2022,changLatentStylebasedQuantum2024} consider a similar structure to our model, except with a style-based encoding, in which the angles in the circuit take the form of trainable transformations $\theta = M\gamma + c$.  The barren plateau proofs in the appendix of \cite{changLatentStylebasedQuantum2024} can be adapted to our case also and the merits and applications of those techniques in comparison to ours are discussed in Appendix \ref{app:variance}. Our contribution to the discussion of hybrid quantum-classical models is to take the argument one step further; even without trainable parameters, the quantum layer can still provide meaningful transformations. \\

The structure of the paper is as follows: in Section \ref{sec:model_structure} we explain the structure of our feature transform $\rho(\gamma)$ and the reasoning behind it. We also extend the transformation into a trainable quantum model. In Section \ref{sec:barren_plateau} we discuss why the feature transformation does not suffer from mode collapse, and show that the trainable model does not suffer from barren plateaus when the feature transform satisfies a condition we call a subvolume law, as well as showing our generative states satisfy this condition. In Section \ref{sec:classical_simulation} we look at two methods for classical simulation, tensor network contraction and Pauli propagation, and argue that our model poses difficulties for them, as well as providing numerical evidence. Finally, in Section \ref{sec:conclusions}, we go through the advantages of our model and provide some questions for further research.

\begin{figure}
    \centering
    \includegraphics[width=\linewidth]{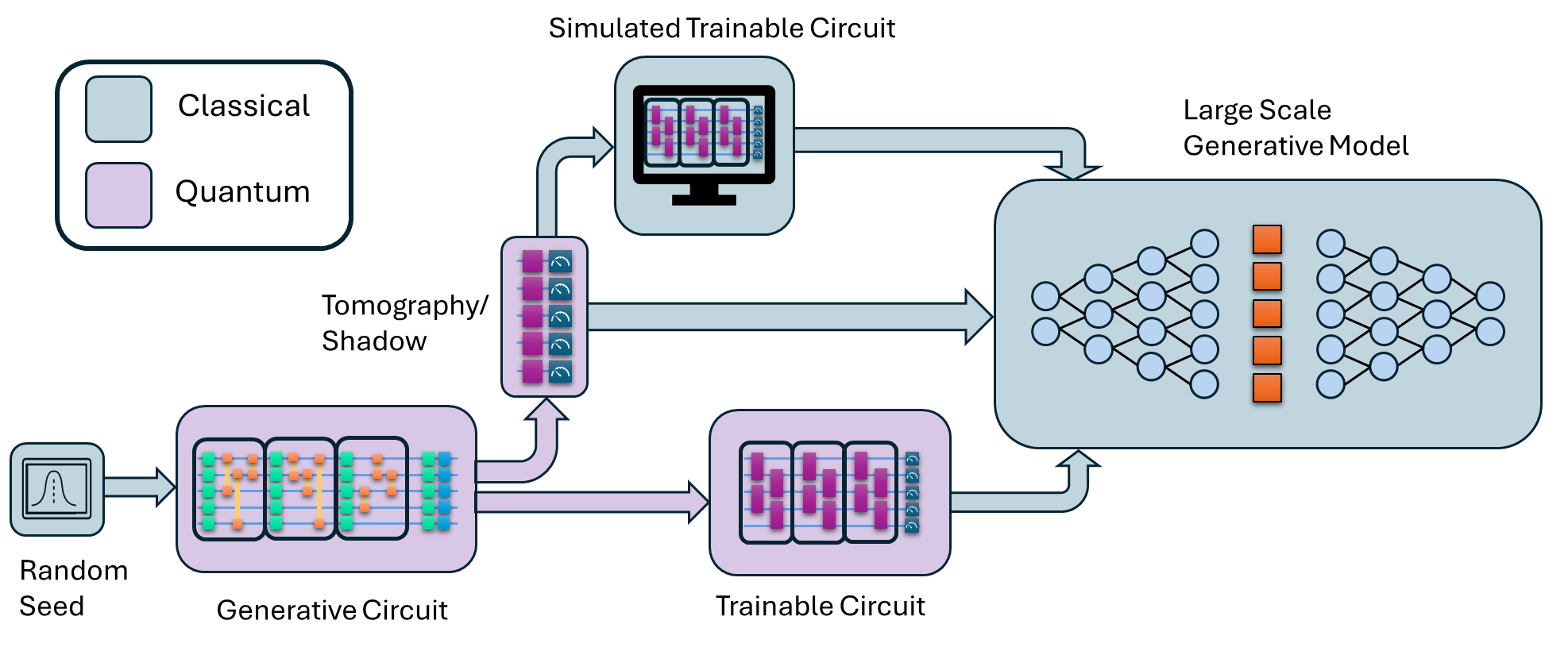}
    \caption{Our proposed use of the generative circuits described in the main text. After the classical random seed is transformed by the generative circuit, the resulting state is then passed through a trainable layer. Alternatively, a classical representation is obtained, which may additionally be trained via a classical simulation. The final step is to feed the transformed information into a large generative model e.g. a GAN, which can scale up to realistic problem sizes.}
    \label{fig:Model_Structure}
\end{figure}

\section{Model Structure}\label{sec:model_structure}
\subsection{Model Outline}
 In our generative model, the quantum circuit has its parameters randomly chosen with each run. The features of the model consist of freely chosen $k$-local observables $O_i$, such as Pauli strings, transforming the probability distribution of the parameters into a different distribution of data. Whilst, at this point, it would be possible to directly feed the results of this transformation to some classical machine learning model, we additionally consider applying another quantum circuit (the \emph{trainable circuit}) with trainable parameters to the output state of initial circuit (the \emph{generative circuit}). If we wish to use a shadow protocol in order to train the parameters of the trainable circuit classically, the trainable part can be absorbed into the observables, as shown in Figure \ref{fig:Model_Structure}.\\
 
When choosing the structure of the generative circuit, we need to ensure that the states produced by it can be locally distinguished from each other (avoiding mode collapse). For the trainable circuit, we additionally need to avoid barren plateaus. The choices we make are known for avoiding barren plateaus (\cite{Zhang_Gauss_2022,Park2023,Wang2023,cerezo2021cost,Khatri2019,Zhao2021,Liu2022,Letcher2023,Basheer2022,Suzuki2023}): for the generative circuit we choose a form of small-angle initialisation and  for the trainable circuit a shallow hardware efficient ansatz (HEA).

\subsection{The Generative Circuit}\label{subsec:generative}
We choose our generative circuit structure to have two imporant properties. Firstly, the distribution $\mathrm{P}_{\Gamma}[\rho=\rho(\gamma)]$ should avoid mode collapse by having large enough variance of local observables Tr($O_i\rho(\gamma)$). In Section \ref{subsec:area_defs} we shall see that this is actually sufficient for our choice of trainable model to avoid barren plateaus. 
Secondly, the reduced density matrices of $\rho(\gamma)$ should (with high probability) not be possible to be approximated classically efficiently. The second condition comes as a consequence of our trainable operators $V^\dagger(\theta)O_i V(\theta)$ being classically simulable; if we could additionally approximate $\rho(\gamma)$ efficiently classically on the support of each operator, our entire model could be evaluated classically, preventing any possible quantum advantage.\\

In addition to the above requirements, we also aim to make our circuit as amenable as possible to execution on real quantum hardware, by choosing them to have sublinear depth with high probability. The structure we choose for our generative circuit consists of $L = O(\log n)$ layers, each of which consists of a round of single qubit $R_X$  rotations with random angles $\gamma_{l,j}$ independently drawn from a Gaussian distribution of constant variance $\tau^2$, followed by CZ gates applied between some randomly chosen pairs of qubits. After these layers, a final round of $R_X$ and then $R_Y$ rotations are applied to each qubit . This is the same structure as studied in \cite{Zhang_Gauss_2022}. In principle we could choose any entangling layer diagonal in the Pauli Z basis without compromising our guarantee of avoiding mode collapse, but we select them to be CZ gates for the ease of analysis and implementation.\\

\begin{figure}
\includegraphics[scale=0.8]{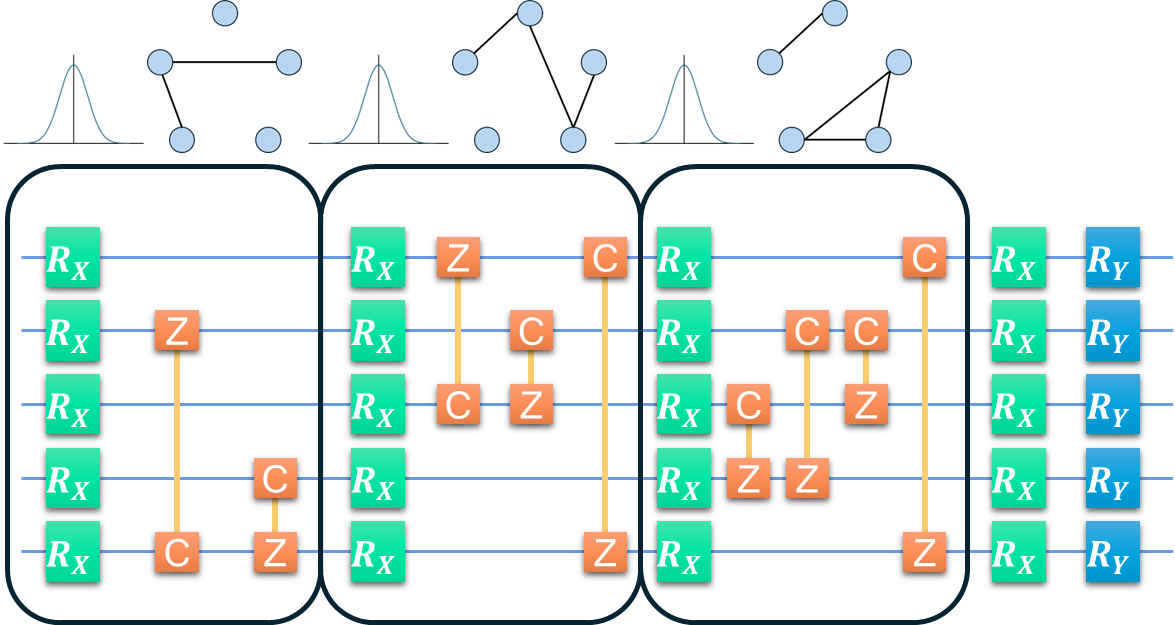}
\caption{Our proposal for a generative circuit, introduced in Section \ref{subsec:generative}, has $L$ layers ($L=3$ pictured here), each consisting of parametrized $X$-rotations on each qubit followed by CZ gates between randomly chosen pairs of qubits. At the end, both parametrized $X$- and $Y$-rotations are applied on each qubit. Parameters are chosen according to a constant variance distribution, whilst the $CZ$ gates are chosen by a $G(n,\log(n)/n)$ Erdős–Rényi graph.}
\label{fig:generative_part}
\end{figure}

The properties of the states created by this model will depend heavily on how the variance $\tau^2$ scales with the number of qubits $n$ and how the distribution determining the placement of CZ gates is defined; in order to satisfy our requirements above, we will choose the variance $\tau^2$ of the RX rotation angles $\gamma_{l,j}$ to be a constant independent of $n$ (see Appendix \ref{app:variance} for a detailed explanation of why this is possible), and the CZ gates in each layer to be drawn according to a $G(n,p)$ Erdős–Rényi graph with $p=\log(n)/n$. In Sections \ref{sec:barren_plateau} and \ref{sec:classical_simulation} we describe the justification for these choices in more detail, but we provide the intuition for them here briefly.\\

The variance is chosen to be constant size, which still avoids mode collapse as well as ensures that the distribution of states is polynomially distinguishable from the maximally mixed state on a subset of $\abs{\Lambda}=O(\log n)$ qubits, which in turn prevents barren plateaus for our trainable observables in Section \ref{subsec:Trainable}. It also means, however, that the values of $\sin(\gamma_{l,j})$ remain large enough to prevent one classical method of simulation, Pauli propagation, to efficiently estimate the state $\rho(\gamma)_{\Lambda}$. This is because this algorithm propagates the observable through the Heisenberg evolution of the circuit in the Pauli basis, dropping terms with sufficiently small weight. 
There are multiple ways to choose the weights to drop, but many of the choices depend on the parameters of the circuit scaling down as the size increases in order to be provably effective. For example, if the low weight is considered to arise from the build up of $\sin(\gamma_{l,j})$ terms when decomposing Pauli rotations, larger angles mean that the truncation error is more significant. This is because having $k$ factors of $\sin(\gamma_{l,j})$ makes the size of the coefficient $(c_\tau)^k$ times smaller on average, where $c_\tau < 1$ is a constant depending on $\tau$. In order to guarantee polynomial precision for constant sized parameters, the truncation methods are forced to keep track of a number of terms growing faster than polynomially (in the case of sine coefficients, the cutoff $k$ scales $\sim \log n$), and thus need superpolynomial runtime $O(n^{\log \varepsilon^{-1}})$.

The layout of CZ gates does not influence the provable absence of mode collapse. Instead it prevents efficient classical approximation via tensor network techniques \cite{Markov_Shi_2008}. Along with Pauli propagation, tensor networks belong to the most effective classical techniques for small-angle circuits.
Our CZ gate layout is chosen such that the light cone of an observable is of macroscopic size with high probability; and that the connectivity graph captured by that light cone has a treewidth of $\Theta(n)$. This prevents efficient simulation. Furthermore the individual CZ layers are sufficiently sparse such that with high probability, they can be transpiled onto a square grid architecture with depth $O(\sqrt{n}\log n)$ \cite{bremner2017achieving}. This sublinear scaling is promising for experiments run on NISQ devices. Our choice is a compromise between having denser CZ graphs to make tensor network contraction harder and sparser graphs being easier to implement physically, and other choices ensuring high treewidth are possible.\\

These choices ensure that our generative circuit gives sufficient detectable variation to the trainable circuit to be of use in a generative model; whilst remaining robust against the two best-known classical simulation techniques. This makes it a prime candidate to provide a bias to generative models which cannot be efficiently reproduced classically.

\subsection{The Trainable Circuit}\label{subsec:Trainable}
As stated before, the addition of a trainable circuit to the model should be considered optional, and more akin to an additional ``quantum inspired" layer, even if quantumly implemented.
Our main concern is to choose for our trainable circuit family a $V(\theta)$ that does not suffer from the problem of barren plateaus.  As outlined in \cite{Lie_Algebra_24,Lie_Algebra_24_2}, this is not merely a problem of choosing the right structure of $V(\theta)$, but requires choosing suitable operators $O_i$, as well the input states $\rho(\gamma)$ from our generative circuit.\\

In addition to an absence of barren plateaus, we choose our observables to be evaluable classically (given a suitable classical description of the input state $\rho(\gamma)$), hence us stressing that the contribution of the trainable component of the model should be interpreted as quantum inspired, rather than offering additional non-classicality. Being able to evaluate the observables classically is not a property necessary for a useful quantum generative model, but rather motivated by the observation from \cite{cerezo2024} that the absence of barren plateaus often leads to such classical simulation, and the pragmatic view that classical training can take advantage of noiseless calculations, parallelism and back-propagation. \\

\begin{figure}
\includegraphics[width=\linewidth]{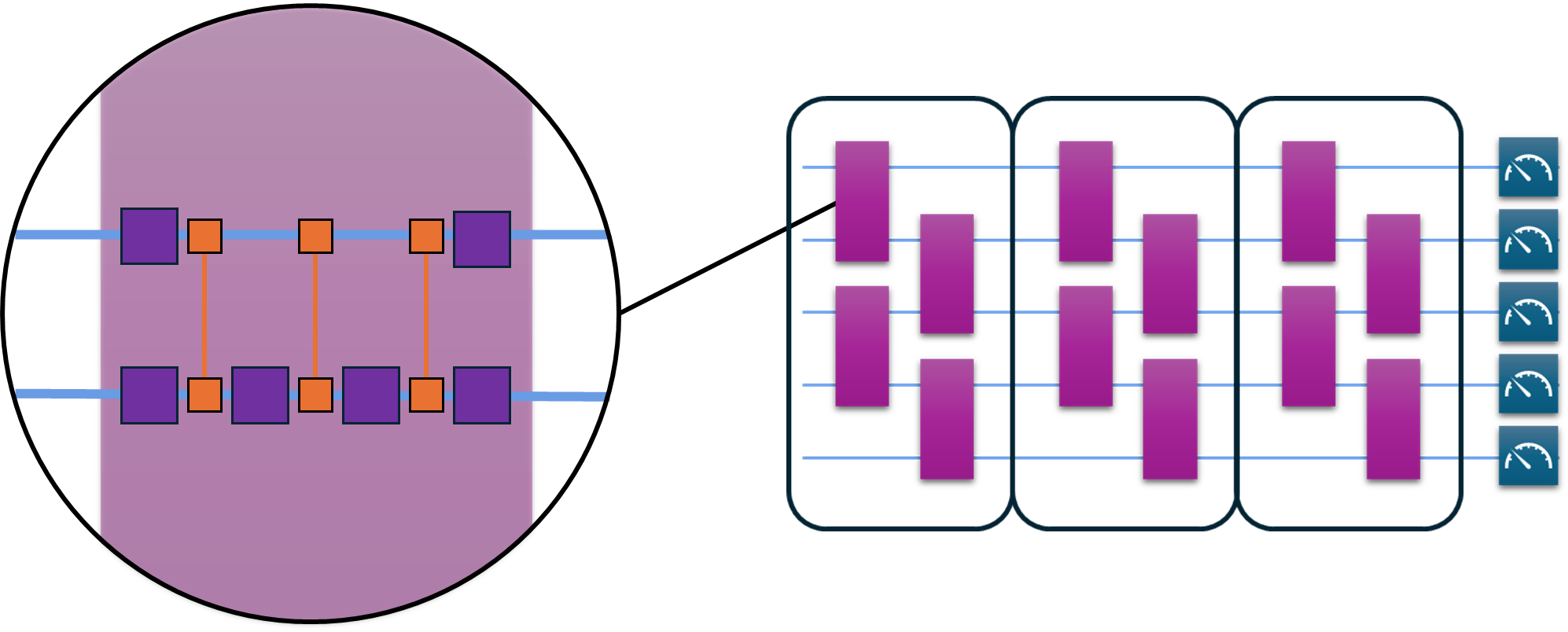}
\caption{The trainable circuit within our model (see Section \ref{subsec:Trainable}) is composed of general parametrized 2-qubit gates arranged in an alternating pattern. These gates can be decomposed further to suit the hardware and/or classical simulation method (a stylised example is shown inside the cutout).}
\label{fig:trainable_part}
\end{figure}

We will choose our observables $O_i$ to be $k$-local operators, and $V(\theta)$ to be an HEA of depth $O(\log n)$. Specifically, it consists of alternating layers of fully-parameterised two-qubit unitaries on neighbouring pairs of qubits on a 1d chain (the ``brick layer ansatz"), the structure considered in e.g. \cite{cerezo2021cost, Leone_HEA_2024}. With this ansatz, one can follow the light cone of a Heisenberg evolved operator $V^\dagger(\theta)O_i V(\theta)$, and see that it can only act non-trivially on at most $O(k\log n)$ qubits, and is thus an element of a real Hilbert space of dimension $4^{O(k\log n)} = O(n^{Ck})$ for a constant $C$ dependent on the exact choice of system size. This element can be explicitly represented classically and, when $k$ is fixed, gives a  polynomial (and thus classically efficient) description of our operator. If we can obtain a classical description of $\rho(\gamma)$ on the relevant $O(k\log n)$ qubits, then we can update our $\theta$ classically. This is exactly the use case that classical shadows \cite{HuangShadow2020} are useful for \cite{Basheer2022}.\\

\section{Avoiding Mode Collapse and Barren Plateaus}\label{sec:barren_plateau}

\subsection{Mode Collapse as Flatness of Distribution}

One key property that we want from the generative circuits $U(\gamma)$, illustrated in Figure \ref{fig:generative_part}, is that the continuous outputs they produce are distinguishable from each other. Otherwise, we run the risk of producing the same feature vector regardless of the random parameters chosen. If then fed into a generative model, this would lead to a problem known as mode collapse, where the model learns to produce a single (or small number of) high quality data, rather than a diverse distribution. This would be due to to our input randomness being ``lost" by our quantum circuit, rather than meaningfully transformed.\\

Since we want our quantum generator to be efficient, we  limit ourselves to a polynomial number of shots. If the observables are exponentially concentrated around their mean over the distribution of $\gamma\sim \Gamma$, then with high probability a polynomial number of measurements cannot differentiate between two data points and the model output is effectively constant, i.e. mode collapse. The only randomness we obtain then comes directly from the measurements (shot noise). Since this is effectively uncorrelated Gaussian noise, we could have obtained the same randomness through much simpler classical methods.  \\

Formally, we can define our instance of mode collapse through a family of output function $f_n(\gamma):=\mathrm{Tr}[O(\gamma)\rho_n]$ for each circuit size $n$, with the input $\gamma\sim \Gamma$ being the classical randomness that is fed into the model.\\

\begin{definition} A function family $f = f_n(\gamma)$ suffers from mode collapse if
\[
\mathrm{Var}_{\gamma \sim \Gamma}[f_n(\gamma)] \in o(n^{-k})
\]
for all $k$.
\end{definition}
This definition is effectively stating that our chosen function concentrates at a faster-than-polynomial rate, and should be thought of as  directly akin to a barren plateau caused by random initialisation of trainable parameters.\\

Preventing this mode collapse is a necessary challenge for a continuous quantum generative model to overcome, as it prevents gaining any advantage from randomly chosen parameters. We are able to show that the  generative circuit we propose avoids this mode collapse, i.e. has polynomially large variance, for local observables. The proof of this can be found in Appendix \ref{app:variance}.\\

\begin{theorem}\label{th:main}
Let $O$ be a $k$-local observable. Then the variance of $O$ is polynomially large over our choice of generative circuit, meaning that 
\[
\mathrm{Var}_{\gamma \sim \Gamma}[ \mathrm{Tr}(OU(\gamma)\rho_0U(\gamma)^\dagger)] \geq \frac{1}{\beta(n)},
\]
where the exponent of $\beta(n)\in\mathrm{poly}(n)$ scales with respect to the locality $k$ as $O(k)$.
\end{theorem}

The key idea behind the proof is analysis of the circuit in the Heisenberg picture i.e. by applying gates to the observable, and noticing that the leading term becomes essentially the product of $kL$ cosines. The proof is based on the approach from \cite{Zhang_Gauss_2022} that was previously refined in \cite{changLatentStylebasedQuantum2024}. In our proof, the only assumption that is required from the entangling gates is that they are diagonal in the $Z$-basis, meaning one obtains not only our model but a large family of circuits that all avoid mode collapse. For other choices of depth $L$, one needs to set the variance to be $\tau^2 = O (\log n / L)$, rather than constant, for circuits with $n$ qubits in order for our proof to work.\\

At this point we pause to note that avoiding mode collapse is only one of the necessary steps for the quantum continuous generative model to be useful. The next necessity, which is that the circuits should not be easy to classically simulate, is the focus of Section \ref{sec:classical_simulation}. Even then however, it may be possible that the feature vector distribution $\vec{f}(\gamma)\sim \Gamma$ is more easily obtained via classical means (something we would call ``classical surrogation"). Even if this is not possible, one would still need to find a dataset for which our quantum feature distribution has a favourable bias. These issues are outside the scope of this work, and there exist plenty of variations of quantum generative models to stress test. These variations can be done by substituting different entangling layers as discussed above, but also through ideas such as correlating the random angles. This is a technique that has seen success in both increasing expressivity (data reuploading) \cite{PerezSalinas2020datareuploading,schuld21}, and barren plateau prevention (correlated initialisations) \cite{Grant2019initialization,Volkoff_2021}.


\subsection{Connection between Mode Collapse and Barren Plateaus}\label{subsec:area_defs}

Our definition for mode collapse strongly resembles the definition of a barren plateau, which happens when
\[
\mathbb{E}_{\theta\sim \Theta} [ |\nabla f(\theta)|] \in o(n^{-k}).
\]
In fact this connection can be made more explicit since, for a general class of widely-used cost functions, the above condition is equivalent to an exponential concentration of the loss function itself \cite{Arrasmith_2022},
\[
\mathrm{Var}_{\theta\sim \Theta} [ f(\theta)] \in o(n^{-k}).
\] 
This is not coincidental; what we call mode collapse is essentially a generative version of a barren plateau: in both cases, we cannot access the relevant values up to needed precision efficiently due to inherent randomness of measurement.\\



In the scenario where we additionally wish to use the states produced by the generative circuit as initial states for our trainable circuit, we need to ensure that those states do not induce barren plateaus in the final output. Since we are not considering here what further happens after the trainable circuit (i.e. a full pipeline as shown in Figure \ref{fig:Model_Structure}), we limit ourselves to preventing barren plateaus in the features produced by the trainable layer, which are expectation values of local observables. As we chose to consider a shallow 1D hardware efficient ansatz, from \cite{Leone_HEA_2024} we know that our generative states should satisfy a condition that is typically called an area law in barren plateau literature, but we feel would be more aptly named a ``subvolume law", since it is strictly weaker than the well-known area law in many-body physics. We discuss this more in Appendix \ref{app:barren}.\\

Fortunately for us, as a consequence of avoiding mode collapse, we also obtain that the states produced by our generative circuit do satisfy the subvolume law, making them suitable initial states for our choice of trainable circuit. This is due to the fact that the local observables (on the input state) have inverse polynomially large variance, meaning those same variables can be used to locally distinguish the state from the maximally mixed state. This is the key subvolume law property needed to prevent a trainable barren plateau - we go into more details about this in Appendix \ref{app:barren}.\\

Although a prevention of mode collapse is sufficient to satisfy the subvolume law, and thus prevent a barren plateau, having subvolume law input states is not a guarantee of mode collapse prevention. This is because the input state could be exponentially concentrated around some fixed, non-maximally mixed state.\\

Finally, we acknowledge that this analysis does not consider another important cause of barren plateaus, noise in the quantum circuit. This issue can be prevented by offloading the training to a classical computer, although hardware noise could still lead to mode collapse.

\section{Classical simulation methods}\label{sec:classical_simulation}

In order to analyse the effectiveness of classical simulation methods, we should first establish the performance of an ideal quantum computer in estimating the observables. Running a circuit with $n$ qubits and $O(\mathrm{poly}(n))$ gates takes $O(\mathrm{poly}(n))$ time, whilst estimating the expectation of a single observable up to precision $\varepsilon$ with probability $1-\delta$ requires $O(\varepsilon^{-2}\log \delta^{-1})$ measurements. That means that if a classical algorithm is able, with probability $1-\delta$, to find an estimate of an observable for that same circuit with precision $\varepsilon$ in time $O(\mathrm{poly}(n,\varepsilon^{-1},\log \delta^{-1}))$, then from a complexity theoretical perspective there is no quantum advantage to be found in the model. In practice, the exact polynomial scaling and pre-factors may play a significant role.\\
 
In this section, we discuss two methods for classically simulating quantum circuits, Pauli propagation \cite{PhysRevA.99.062337,shao2024simulating,angrisani2024classically} and tensor networks \cite{PhysRevLett.91.147902,cirac2009tensor,Markov_Shi_2008}. From our understanding of the literature, these are the two methods most applicable to our generative circuits, but we acknowledge there exists a wide suite of approaches \cite{XU20254104}. Even for the two methods we choose to consider, there exists a rich body of literature, and we have chosen specific implementations of these methods we feel representative of the overall difficulty, rather than what may prove to be the absolute optimal approach. We encourage the community to challenge our assertions, in order to situate the border of classical power as clearly as possible. \\

The following two subsections are each dedicated to a particular simulation method. We begin with a brief description of the simulation approach, and then provide the intuition behind how we have chosen our generative quantum model, in order to be difficult for that technique. We also provide numerical simulations which compare the exact observable values (features), obtained using statevector simulation, with those predicted using classical simulation methods.
For the numerical simulations, we randomly generated 100 instances of generative circuits of sizes 5, 10, 15, 20 and 25 qubits, using angle variance $\tau^2 = 1/9$. The observables we consider are single qubit Z-observables, and for each circuit instance we focus on the maximum error over all features. This is because we demand of our classical simulation technique the ability to approximate any chosen observable efficiently, which is not necessarily the case if the average error over all sites (features) is small.\\

Briefly summarising the results, we find that the limitations of tensor network methods are quite easy to demonstrate numerically, whereas the results for Pauli propagation are more subtle. Nevertheless, the evidence points to polynomial time scaling methods being unable to increase their precision beyond a given threshold, in contrast to the ideal quantum computer. Since we are limited in size by the statevector simulation, we believe it would be interesting to extend this study by finding a related circuit class which can be easily simulated with tensor network approaches, but whose Pauli propagation difficulty still reflects that of our generative model.

\subsection{Pauli Propagation}\label{subsec:pauli_prop}

In order to understand how we choose our circuit to be difficult for Pauli propagation, we first outline how it works. Pauli propagation is based on the idea of tracking observables written as Pauli sums $O = \sum_j c_j P_j$ through the Heisenberg picture of the circuit i.e. applying gates to the observable. In each step of the algorithm, the current Pauli sum is conjugated by the current gate, yielding the Pauli sum for the next step. For a Clifford gate $C$, all Pauli strings are mapped to Pauli strings via $C^{\dagger}PC = P'$. For a Pauli rotation $R(\gamma)$, all Pauli strings that commute with it stay unchanged, while those Pauli strings $P$ that do not commute are mapped to $R(\gamma)^\dagger PR(\gamma) = \cos(\gamma)P + \sin(\gamma)P'$. This produces two new strings to keep track of.\\

Though the number of terms in the sum grows exponentially, an approximation algorithm obtained by truncating terms can cause the algorithm to run in polynomial time, at the cost of guaranteeing an accurate solution. Truncation choices include dropping the terms with coefficients of small absolute value or ones that have a large number of non-identities in their Pauli string. Dropping small coefficients can lead to the total error being small when e.g. the angles of the circuit are close to zero \cite{Lerch2024}. In practice, multiple truncation strategies are often used in combination, reducing the time and memory overhead of the algorithm. Since the truncation thresholds for these combined algorithms are often chosen heuristically, we limit ourselves to studying single truncation approaches. For an overview of truncation strategies, see Appendix \ref{app:Trunc_Methods}.\\

Because Pauli propagation is well-suited to small-angle circuits, our choice of circuit should try to minimize its effectiveness. For a quantum computer, finding the expectation of an observable up to an error which scales inverse polynomially takes polynomial time, and so we set the goal for Pauli propagation to find the expectation in polynomial time with respect to the number of qubits and inverse error jointly. The estimates from \cite{Lerch2024} for the resources taken by Pauli propagation indicate that this would take quasipolynomial runtime $O(n^{O(\log \varepsilon^{-1})})$, and we sketch below the justification for why poor scaling is expected.\\

\begin{figure}
\centering
\textbf{Pauli Propagation Performance}\par\medskip
\begin{subfigure}[b]{0.45\textwidth}
\centering
\includegraphics[width=\textwidth]{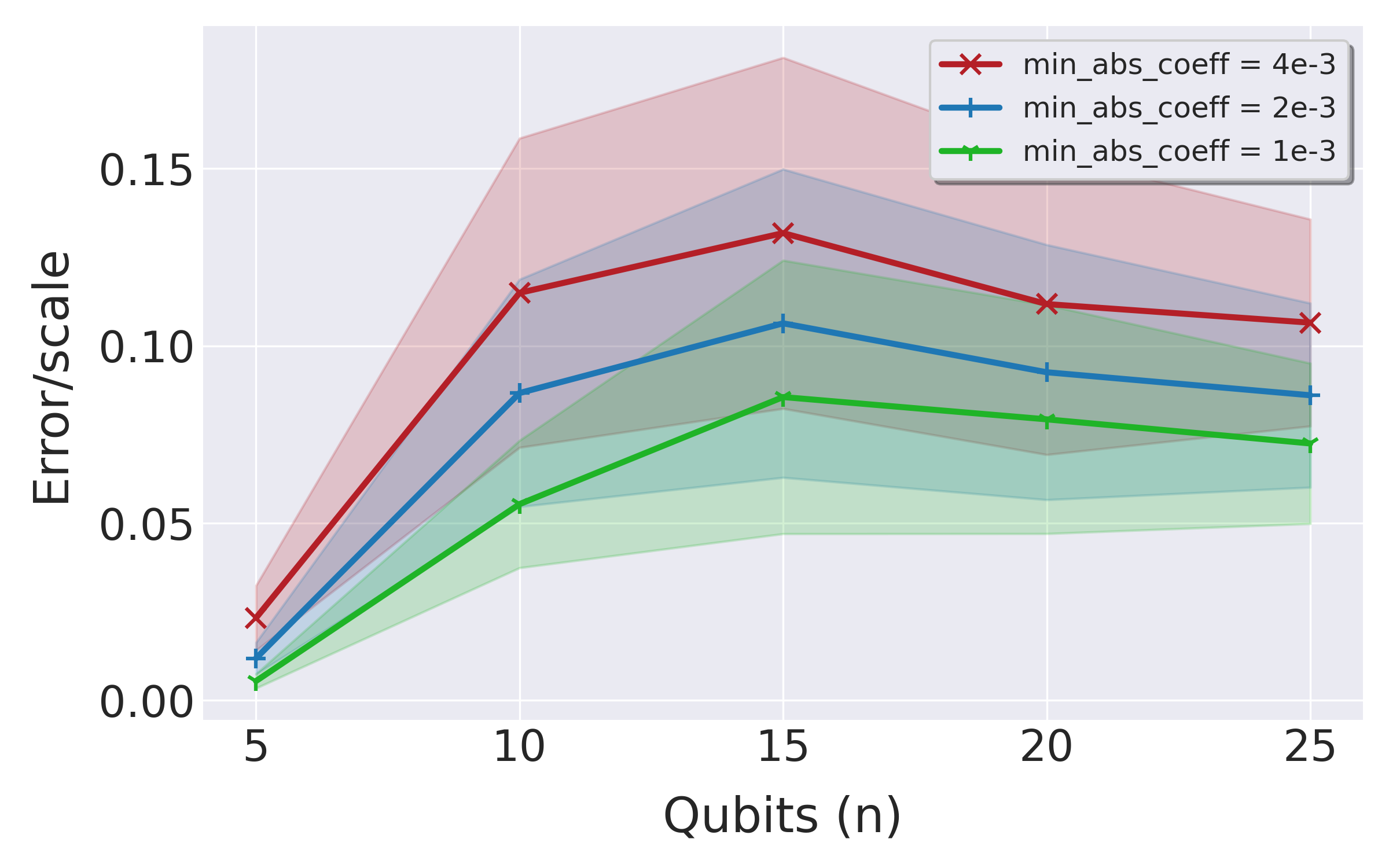}
\caption{Small coefficient truncation}
\label{subfig:pauli_prop_abs}
\end{subfigure}
\hfill
\begin{subfigure}[b]{0.45\textwidth}
\centering
\includegraphics[width=\textwidth]{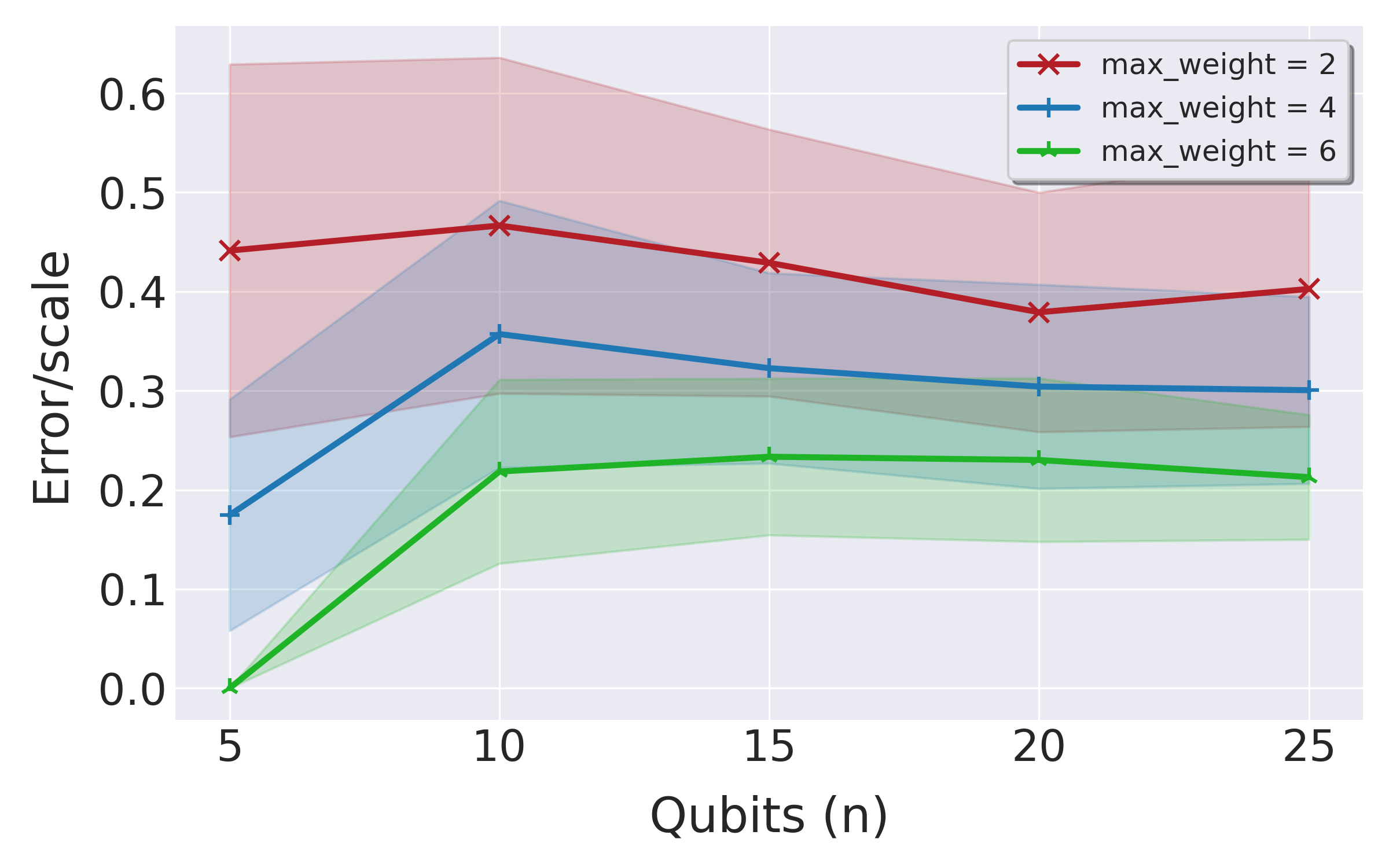}
\caption{Pauli weight truncation}
\label{subfig:pauli_prop_weight}
\end{subfigure}
\\
\begin{subfigure}[b]{0.45\textwidth}
\centering
\includegraphics[width=\textwidth]{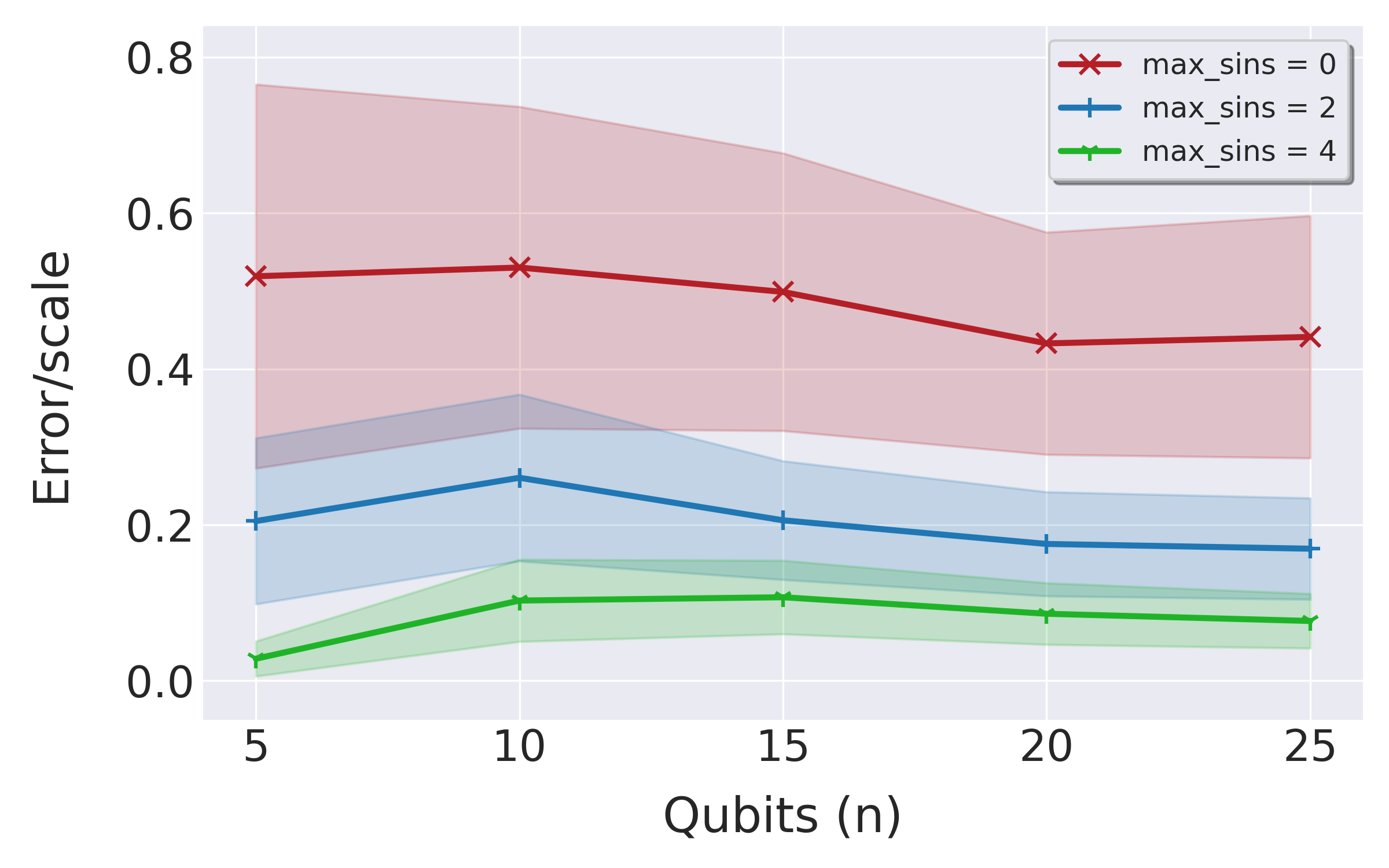}
\caption{Sine truncation}
\label{subfig:pauli_prop_sins}
\end{subfigure}
\hfill
\begin{subfigure}[b]{0.45\textwidth}
\centering
\includegraphics[width=\textwidth]{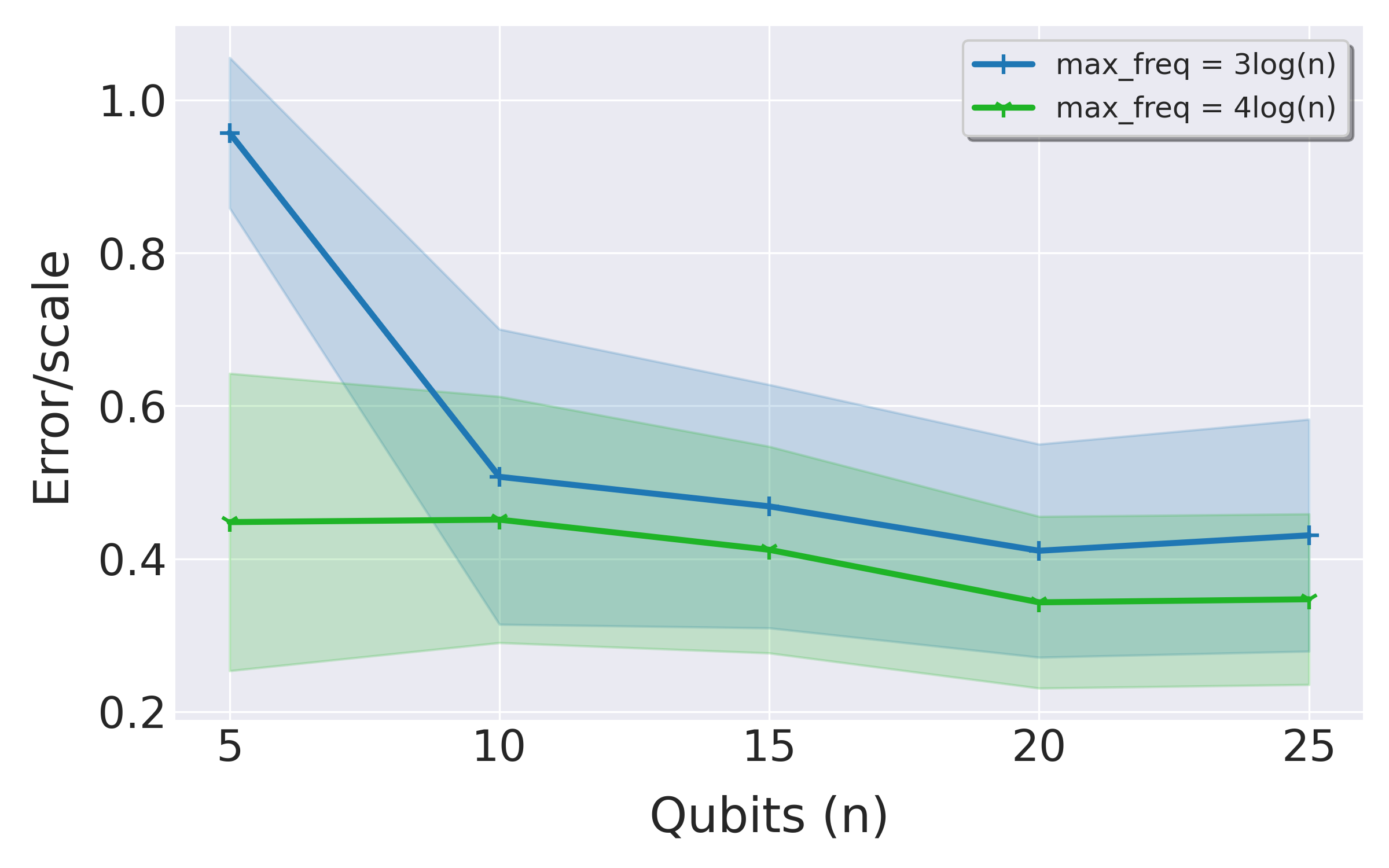}
\caption{Frequency truncation}
\label{subfig:pauli_prop_freq}
\end{subfigure}
\\
\begin{subfigure}[b]{0.55\textwidth}
\centering
\includegraphics[width=\textwidth]{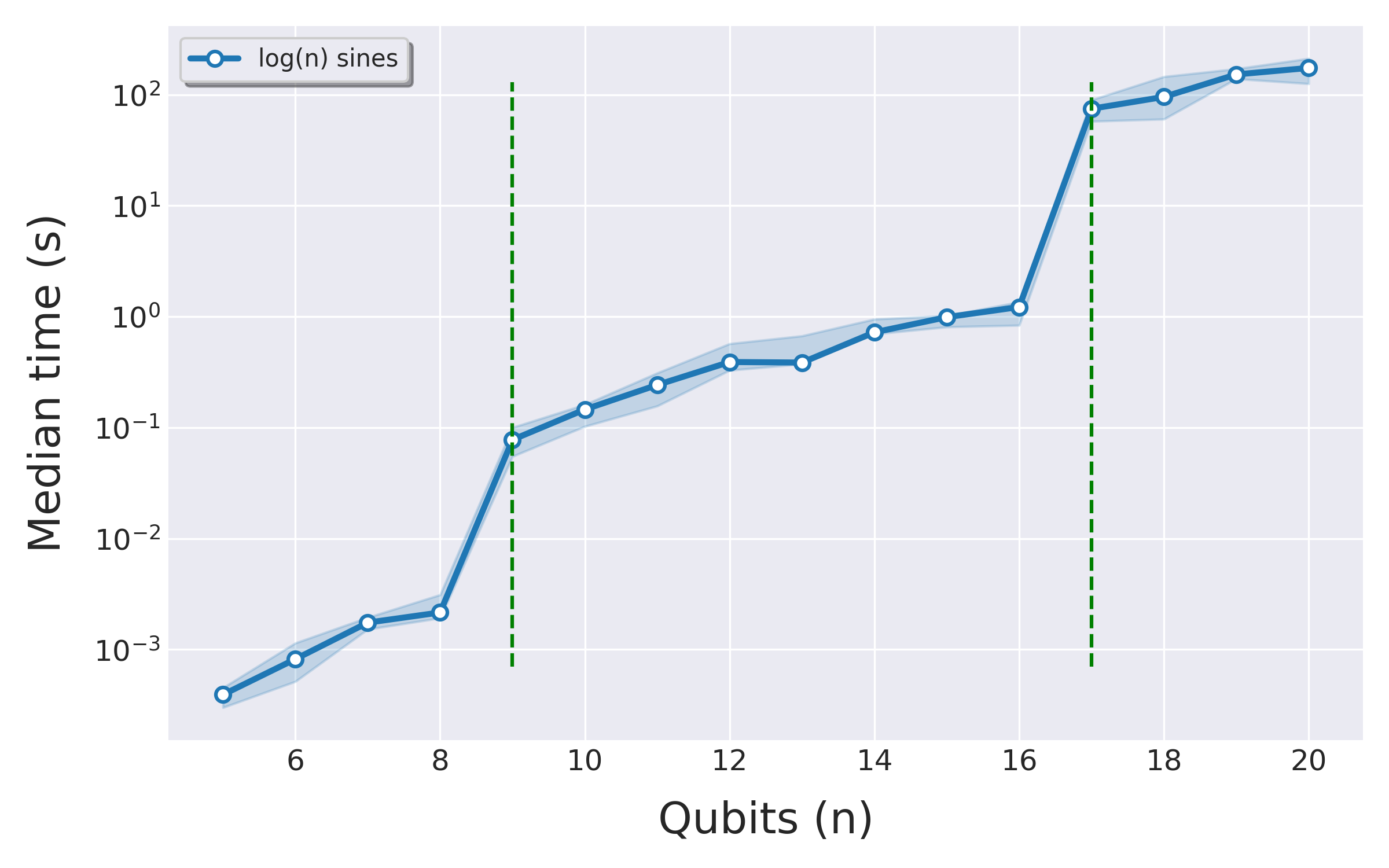}
\caption{Time scaling when increasing cutoff precision with system size.}
\label{subfig:pauli_prop_time}
\end{subfigure}
\caption{Analysis of Pauli propagation approaches to simulating our generative circuits. In \ref{subfig:pauli_prop_abs}-\ref{subfig:pauli_prop_freq}, we present the largest single-qubit Z-observable error for various Pauli propagation truncation methods, scaled proportionally to the average observable magnitude. Each point is averaged over 100 circuit instances, defined according to the generative model in Section \ref{sec:model_structure}, with angle variance $\tau^2 = 1/9$. The shaded regions correspond to one standard deviation of the error. For each method, we plot several truncation thresholds corresponding to an increase in the exponent of the polynomial runtime. The fact that the error relative to the average observable magnitude appears to stabilise implies that any polynomial runtime has a threshold polynomial precision, which it cannot surpass. By contrast, an ideal quantum computer can reach any polynomial precision with a polynomial runtime.\protect\linebreak
\ref{subfig:pauli_prop_time} shows the median runtime on a laptop of the Pauli propagation algorithm with the cutoff at a logarithmic number of sine factors. Each point considers 10 circuits generated in the same manner as those above, with the shaded region showing the area where half of the runtimes lie. The runtime of the algorithm increases rapidly with the number of qubits making this approach impractical for larger instances, even though we expect it to give accurate estimates. The green lines denote where the cutoff increases in integer value, explaining why we observe larger jumps for those instances.}



\label{fig:pauli_prop_comb}
\end{figure}

In our model, the angles of the single qubit rotations are drawn from a Gaussian distribution with constant variance. We chose the constant to be small for the numerical simulations. This ensures that with large probability we have $|\sin(\gamma)| < |\cos(\gamma)|$, and therefore we may keep the small-angle intuition that truncating terms with many sines is a sensible strategy. Intuitively, the size of the truncated term shrinks at the same rate as the overall expectation when keeping only a constant number of sine coefficients, meaning that to obtain better precision requires keeping track of larger products of sine terms, the number of which scales superpolynomially. We show the runtime of this approach for small numbers of qubits in Figure \ref{subfig:pauli_prop_time}, demonstrating the rapid growth we predict. We expect the other most commonly used truncation strategy, ignoring terms with a number of non-identity Pauli operators (the ``Pauli weight") above a given constant, to fail as the system size grows. This is because the entangling layers are likely to create terms with a logarithmic-sized Pauli weight that contribute polynomially to the expectation, which for large enough circuits will be truncated. Keeping track of these logarithmic terms would also lead to superpolynomial scaling.\\\\


In Figure \ref{fig:pauli_prop_comb} we present numerical evidence to support our claim that efficient Pauli propagation techniques cannot match the performance of a quantum computer. We do this by comparing the estimates of observables obtained by Pauli propagation, as implemented in the Julia library presented in \cite{rudolph2025pauli}, to the exact values obtained by statevector simulation for small numbers of qubits. We consider four different truncation strategies. Figure \ref{subfig:pauli_prop_abs} shows small coefficient truncation, which truncates terms when the absolute value of their coefficients is smaller than a specified threshold. Figure \ref{subfig:pauli_prop_weight} considers Pauli weight truncation, which truncates the Pauli terms once they have more non-identities than a given threshold. Figure \ref{subfig:pauli_prop_sins} shows the results for sine truncation, which truncates the terms that have accumulated a large number of sine factors from Pauli rotations. The final approach we test is frequency truncation, shown in Figure \ref{subfig:pauli_prop_freq}, where we truncate terms that have accumulated a large number of factors, either sine or cosine, arising from non-commuting Pauli rotations.\\

As it would be sufficient to show difficulty to have on average at least one difficult observable per circuit, we plotted the maximum error of all single qubit Z-observables for each instance, averaged over multiple circuits of the same size. As we show in Appendix \ref{app:variance}, the average size of these observables decreases with the number of qubits, and therefore we consider the error relative to the size of the observables by dividing with a factor determined by the scale of the expectation value. This is a necessary step, since reporting only the additive error can lead to the false impression of improving performance as circuit size increases (think of an algorithm which always predicts a value of 0 - as the circuit size increases, this algorithm also improves its additive error). The choice of scale factor $e^{-L\tau^2/2}$ comes from the fact that $e^{L\tau^2/2}\mathbb{E}\left[\mathrm{Tr}(Z_j U(\gamma)\rho_0 U(\gamma)^\dagger)\right]$ is constant (see Appendix \ref{app:variance}), and therefore $e^{-L\tau^2/2}$ represents the ``expected" scaling of our observables.\\

For all truncation strategies, we observe for small numbers of qubits that the errors seem to scale proportionally to our observables, which suggests that in order to obtain inverse polynomially scaling errors beyond $e^{-L\tau^2/2} \sim n^{- \tau^2 \log e/2}$, we need to scale the truncation threshold in a way that makes the algorithm run in superpolynomial time. This contradicts the requirements we imposed for an efficient classical algorithm. In principle one can choose any fixed locality, and include X and Y terms, to create features, but we limit our numerics to single qubit Z-observables since lower-locality observables have higher expected magnitudes, and variables involving X and Y  are expected to have qualitatively similar behaviour.\\

We leave open the possibility that other truncation approaches may be more suitable here; 
we suspect that an algorithm which scales the truncation threshold as the propagation evolves, increasing the number of kept terms near the end, would be more efficient. We also note that in practice, quasipolynomial scaling may be sufficient. We leave it as a challenge to the community to determine if the scaling of a quantum computer can be matched.

\subsection{Tensor Network Simulation}\label{subsec:Tensor Networks}

Approximate approaches using tensor networks such as matrix product states and operators (MPS/MPO) and projected entangled pair states (PEPS) have had great success in a wide range of simulations of quantum systems. For an introduction into these methods, see \cite{Bridgeman_2017}. These methods produce efficient approximations when entanglement of the system is limited, particularly in the case of the area law states seen in many-body physics. Since our proof of an absence of a barren plateau relies on limited subsystem entanglement, it is important to consider whether these tensor network techniques are suitable for simulating our model.\\

From the perspective of tensor network simulation, the tensor network to contract corresponds to the expectation value  $\bra{\Phi(\gamma)}\tilde{O}\ket{\Phi(\gamma)}$, where $\ket{\Phi(\gamma)}=U(\gamma)\ket{0}^{\otimes n}$ and $\tilde{O}$ is an observable with non-trivial support on $O(k\log n)$ qubits (either $V^{\dagger}(\theta)O_iV(\theta)$, or one of the $O(k\log n)$ length Pauli strings to describe the relevant reduced state, depending on how the generative circuit's output is to be used). This provides us with a tensor network we need to contract into a single value whose
difficulty is determined by the largest intermediate tensor we must store during contraction. Only gates in $U(\gamma)$ which appear in the light cone of $\tilde{O}$ need be considered, as other gates cancel out during the contraction.\\

As a tree decomposition of the network graph can be used to obtain a contraction order, the best possible complexity of the contraction depends on the treewidth of the graph \cite{Markov_Shi_2008}. We denote by $G_l$ the graph on $n$ vertices corresponding to the qubits, with edges between those vertices where there is a CZ gate between qubits in $l$th layer of the generative circuit. As single qubit gates do not affect the treewidth of the tensor network, the choice of CZ gates according to graphs $G_l$ in each layer completely determine the final treewidth. Moreover, all the gates outside the relevant light cone can be made to cancel out during the contraction by using the fact that CZ gates commute with each other, and therefore we need to reach a high treewidth inside that light cone to ensure difficulty. Our strategy is to choose our graphs $G_l$ such that the following statements are true: 
\begin{itemize}
    \item After going back through layers $l=2\ldots L$, all $n$ qubits are in the light cone with high probability.
    \item The graph $G_1$ has a treewidth of $\Theta(n)$ with high probability. 
\end{itemize}
This will ensure that with high probability $G_1$ is a subgraph of our tensor network, and thus provides a lower bound for the treewidth. Our approach is inspired by \cite{bremner2017achieving}, who consider placement of diagonal 2-qubit gates for a different purpose, relating to the anti-concentration of Instantaneous Quantum Polynomial (IQP) circuits. At each layer we take $G_l$ to be a $G(n,p)$ Erdős–Rényi graph with $p=\log(n)/n$. One can see the expected number of 2-qubit gates acting on a single qubit is $\sim \log n$, with the expected total number of 2-qubit gates in each layer $\sim n \log n$. \\

We now sketch the relevant properties of our circuit that these graphs provide. The first concerns the spread of the light cone. If the expected number of connections for a single node is $c$, then in $L=\log_c n$ layers we expect to have discovered $\Theta(n)$ nodes (at each step we expect to discover $c$ new nodes from each of our old nodes, with duplicates only becoming an issue as we approach macroscopic scales). Thus with $L=\log(n)/\log(\log(n))$ we expect to have seen $\Theta(n)$ of the qubits - we will however take $L=\log(n)$, to increase this further. Figure \ref{fig:Light_Cone} shows the number of qubits in the light cone after $\log(n)-1$ layers, averaged over many graph instances. \\

\begin{figure}
\includegraphics[width=\textwidth]{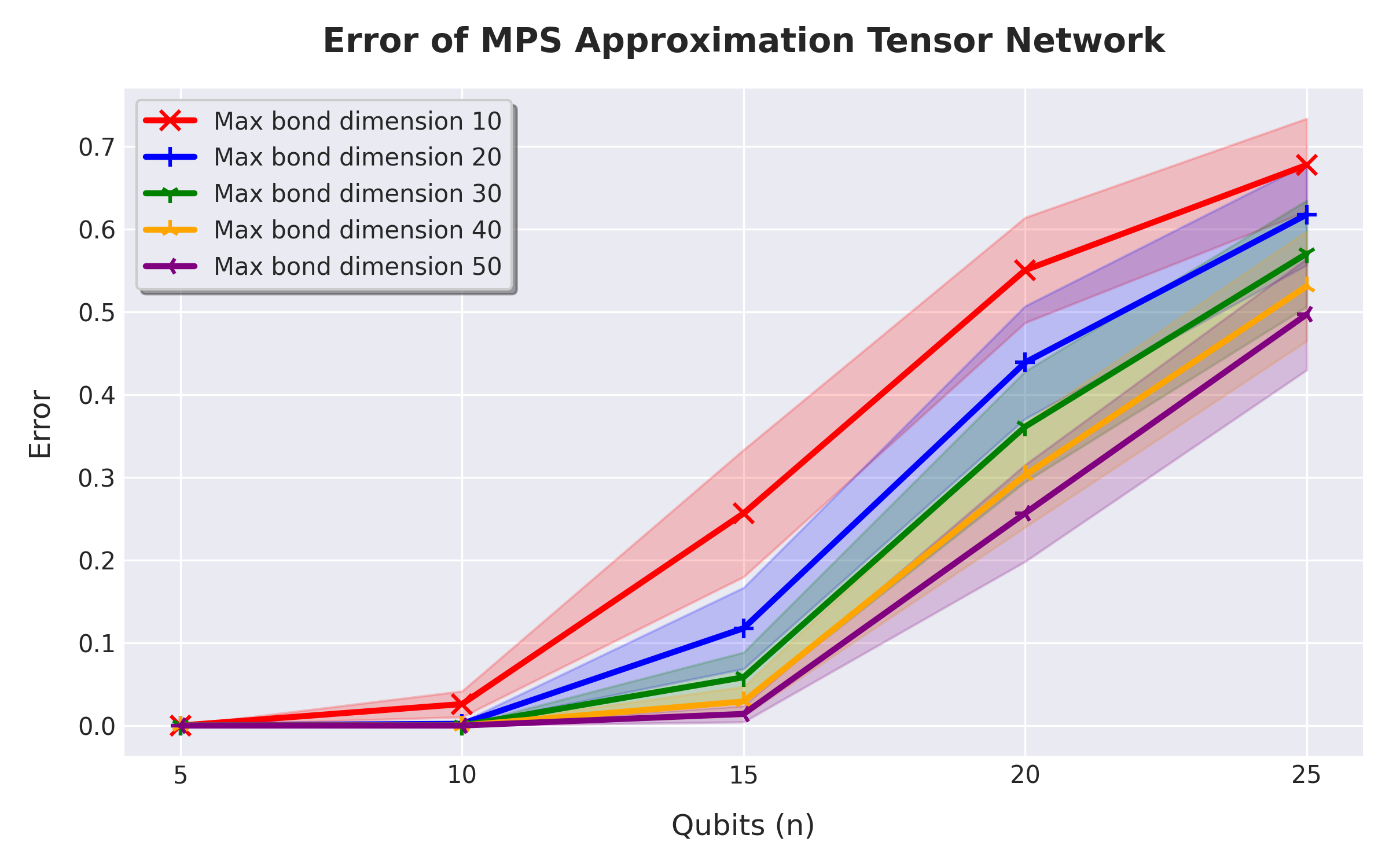}
\caption{Largest single-qubit Z-observable error of the MPS approximation with bounded bond dimension. Each point is averaged over 100 circuit instances, defined according to the generative model in Section \ref{sec:model_structure}. The shaded regions correspond to one standard deviation of the error. Angle variance is chosen to be $\tau^2 = 1/9$. From this plot, we can see that low bond dimension MPS approximations are not effective for our model.}
\label{fig:tensor_network_error}
\end{figure}

Now that our light cone is of the order $\Theta(n)$, we can consider the treewidth $\TW$ of our random graph in the first layer, $G_1$. Here we use a result from \cite{RandomTreeWidth}, which states for $G(n,p)$ graphs that $\TW(G) = n-o(n)$  if $p$ grows faster than $c/n,\, c>1$. Thus we can be sure that the treewidth of our tensor network contraction is at least as large as $\Theta(n)$, ensuring exponential scaling for an exact calculation. The average treewidth of $G_1$, considering only qubits within the light cone, is plotted in Figure \ref{fig:Tree_Width}.\\

In practice, one would try to simulate the circuit with an approximate tensor network method, such as an MPS with fixed bond dimension. We still expect such methods to perform poorly on our circuits as these tensor models exhibit locality properties not respected by our choice of entangling gates. This is supported by numerical simulations that we ran using the quimb library \cite{gray2018quimb}, as illustrated in Figure \ref{fig:tensor_network_error}.\\

We note that, to obtain both these properties (a large light cone and high treewidth), it would have been sufficient to take sparser graphs with  $p=c/n,\, c>1$ instead. However, layers of our chosen form can be transpiled onto a square grid architecture with depth $\sqrt{n}\log(n)$ with high probability; resulting in a total generative circuit depth of $O(\sqrt{n}\log^2(n))$. We feel that the extra factor of $\log(n)$ in the depth is worth the extra difficulty that results in the tensor network contraction, but if one was aiming to use these circuits for a supremacy experiment, it is likely that trading off some of the classical difficulty for improved fidelity on the quantum hardware would be beneficial. Furthermore, the connectivity of $CZ$ gates does not affect the Pauli propagation or the barren plateau arguments, meaning there is a lot of freedom of choice, so long as one ensures difficulty for tensor network methods. \\

\section{Conclusions and further directions}\label{sec:conclusions}

In this work, we have proposed a quantum circuit family to use as a generative seed; by transforming classical randomness in a way that (to the best of our knowledge) is not classically simulable, we hope to provide an initial quantum bias to large-scale classical generative models. Our goal is to encourage exploration of quantum-hybrid models beyond the standard toy datasets, but in a way in which we can be confident the quantum part is truly contributing something beyond classical capabilities. It may seem counter-intuitive to the reader that a quantum layer that is problem agnostic (since it does not contain trainable parameters) could be nevertheless useful; we point out that this is the basis of models such as quantum extreme learning machines \cite{innocenti2023potential,xiong2025fundamental,delorenzis2025behind,ELMResOpp} and reservoir computing \cite{ELMResOpp,PhysRevApplied.8.024030,nakajima2020quantum_reservoir_review,gyurik2025quantum_reservoir_featuremaps}, so this idea is not completely novel.\\

We acknowledge that there remain a number of questions to be answered, both theoretical and practical. On the theoretical side, the question remains whether there exists an improved truncation rule for Pauli propagation beyond the ones considered in this work; we believe an algorithm which scales the truncation threshold as the propagation evolves, increasing the number of kept terms near the end, would be a promising direction for future research. On the practical side, we have not touched on the issue of device noise; something also known to induce barren plateaus (e.g. \cite{Wang_2021}), and make Pauli propagation a viable simulation strategy \cite{PhysRevA.99.062337,shao2024simulating,fontana2025classical}. This shifts the goalposts firmly in the classical simulation's favour. Another interesting question we leave to a future work is to determine whether our generative circuit model can be surrogated by another type of classical model, rather than trying to simulate them directly. For example, one could try to train a neural network to reproduce the mapping $\gamma \rightarrow \langle \vec{O} \rangle_{\rho(\gamma)}$. This kind of adversarial approach is necessary to determine if our idea of minimal quantum contribution is viable. If our circuits pass this test, the next step is to look for datasets for which this bias is well suited.\\

On a more positive note, there is still a lot to explore regarding the structure of the circuits themselves. Anything diagonal in the Z basis can be substituted into the entangling layers without compromising the proofs for avoidance of mode collapse (and therefore also barren plateaus), opening the door to more hardware-friendly gates or higher-order interactions. This also opens the door to non-Clifford gates - something considered with respect to barren plateaus in \cite{changLatentStylebasedQuantum2024}, but whose influence on the Pauli propagation algorithm is less clear.\\

In summary, whilst currently the majority of interest in quantum generative models is focussed on discrete models, in order to exploit the strong results regarding the difficulty of sampling (\cite{bremner2010classical,aaronson2013computational,boixo2018characterizing,RevModPhys.95.035001}), we hope this work shows the question regarding the usefulness of deterministic models is still open. We should not forget that the majority of neural networks deal with real values and continuous mappings \cite{shalev2014understanding,goodfellow2016deep}, and the approach taken here is a lot less resource intensive than that for a binary discretisation. Furthermore, generative features can be generated once and re-used (akin to random numbers for non-cryptographic purposes), and if the trainable layer is done classically after taking classical shadows, then one can exploit all the advantages in training of parallelisation, back-propagation and noiseless gradient evaluations.

\section{Acknowledgements}
We thank Supanut Thanasilp as well as Mario Ponce, Jami Rönkkö, Fedor Šimkovic and Martin Leib for valuable discussions and contributions, and the anonymous referees of QCTiP and TQC for useful feedback on an earlier draft.

\bibliography{references}

\appendix

\section{Tensor Networks: Numerical Studies}
In Section \ref{subsec:Tensor Networks}, we presented results from the literature regarding treewidth of random graphs. To recap, (exact) tensor network simulation methods scale exponentially in the treewidth of the underlying entanglement graph \cite{Markov_Shi_2008}. We also considered numerical tensor approximation with a fixed bond dimension, and observed that the error grows as the circuit size increases. In this appendix we supplement those results with a numerical study of the typical treewidth observed in our circuits.\\

As stated in our main text, we ensure a high treewidth via the following two steps:
\begin{itemize}
\item The entangling layers $l=2\ldots \log n$ ensure a large proportion of the qubits are within the observable's light cone.
\item The final entanglement graph of layer $l=1$ has a high treewidth.
\end{itemize}

In Figure \ref{fig:Light_Cone}, we plot the average fraction of qubits included in the light cone prior to the final (in the Heisenberg picture) generative entanglement layer. For each circuit size, we average over 100 circuit instances. To calculate the size of the final light cone, we begin with a single node (representing a single qubit observable). For each layer, we then generate an Erdős–Rényi graph with $p=\log(n)/n$ (as described in the main text). We then update the light cone to include all nodes connected to our original choice, since in our circuit they are now connected by an entangling $CZ$ gate. We repeat this process $\lceil \log n\rceil$-1 (in this work, all logarithms are base 2) times, and note the final proportion of nodes included in the light cone.\\

Observing the results, we see that already for small systems sizes (where we would still be able to simulate the circuit with statevector approaches), not only do we expect all qubits to be included in the light cone, we expect this to be true with high probability. We can therefore be confident that we are unlikely to underestimate the difficulty of tensor network simulation due to light cone cancellations.\\

\begin{figure}
    \centering
    \includegraphics[width=\linewidth]{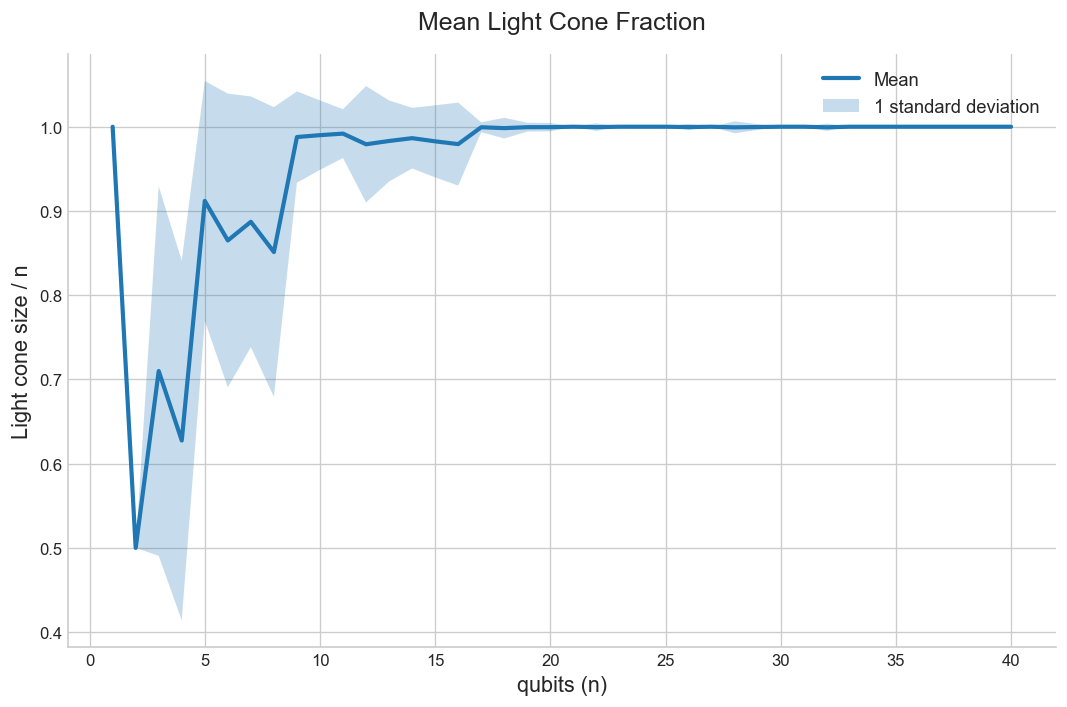}
    \caption{The fraction of qubits included in the reverse light cone for a single qubit observable, after $\lceil \log n \rceil -1$ layers. For each circuit size, we average over 100 instances. We observe that even for small numbers of qubits, all qubits are contained within the light cone with overwhelming probability.}
    \label{fig:Light_Cone}
\end{figure}

\begin{figure}
    \centering
    \includegraphics[width=\linewidth]{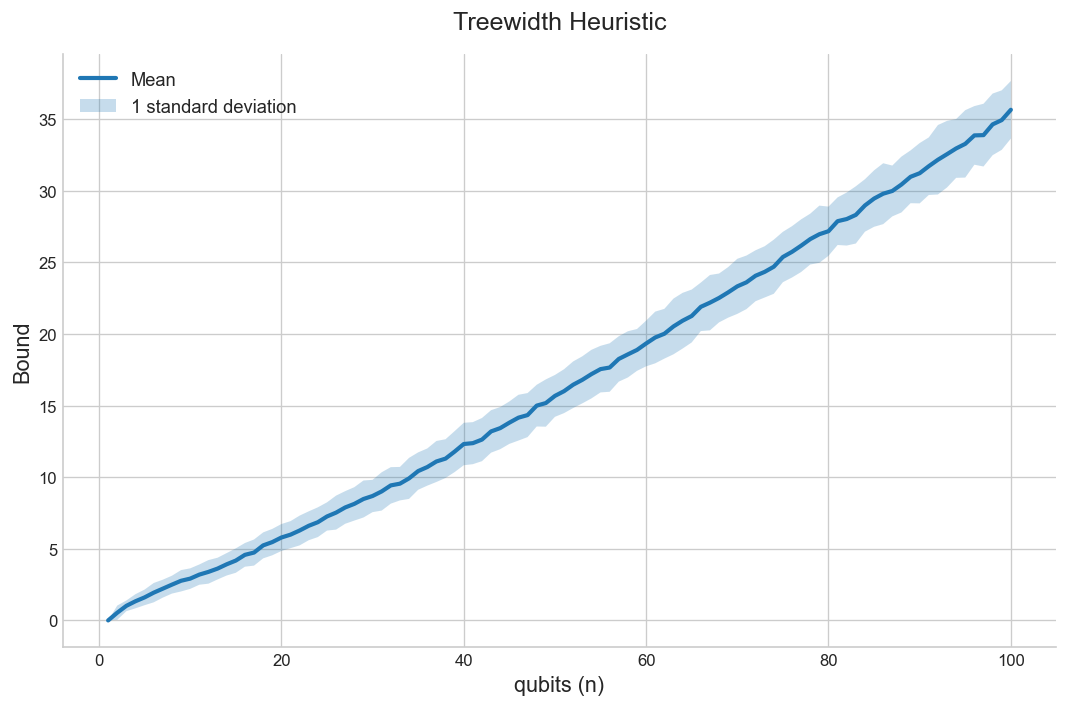} 
    \caption{A heuristic upper bound for the final entanglement graph of our generative circuits, in which connections not in the light cone have been removed. For each circuit size, results have been averaged over 100 instances. Whilst this heuristic is not a rigorous guarantee of difficulty, the runtime of tensor network evaluation algorithms depends on the heuristic tree decomposition found, rather than the true graph value, which is NP-complete to compute.}
    \label{fig:Tree_Width}
\end{figure}

From the light cones generated for Figure \ref{fig:Light_Cone}, we then generate a final entanglement graph for layer $l=1$. We then run the minimum fill-in heuristic function in python's \textsc{NetworkX} package \cite{SciPyProceedings_11}. Whilst this provides an \emph{upper bound} on the true treewidth, and is thus not a guarantee of difficulty, finding a minimal tree decomposition is in general an NP-complete problem. Thus we believe the results of Figure \ref{fig:Tree_Width} to be representative of the difficulty of actual tensor network simulators, which require an explicit tree decomposition to evaluate. We see that the heuristic scales linearly with $n$ (implying an exponential growth for the exact simulation). Combined with the numerics in Figure \ref{fig:tensor_network_error}, we feel confident in concluding tensor network methods are not sufficient for accurate simulation of our generative circuits. 

\section{Lower Bound of the Variance}\label{app:variance}

In this appendix we prove the lower bound for the variance of the generative circuit from the main text. This is exactly the condition that allows the model to avoid mode collapse. Theorem \ref{th:main} is obtained as a special case of Corollary \ref{cor:variance}, where the setting is more general. The approach of the proofs in this section is based on the one used in \cite{Zhang_Gauss_2022} that has been previously tightened in \cite{changLatentStylebasedQuantum2024}. In our assumptions, the entangling gates are assumed to be diagonal in the $Z$-basis as opposed to being composed of CZ gates as in \cite{Zhang_Gauss_2022} or Clifford gates as in \cite{changLatentStylebasedQuantum2024}.\\

Let $U(\gamma) = R_{L+2}(\gamma_{L+2}) R_{L+1}(\gamma_{L+1}) \prod_{l=L}^1 U_l(\gamma_l)$ be the quantum circuit, where $U_l(\gamma_l) = U_l' R_l(\gamma_j) $ is a composition of entangling gates $U_l'$ and Pauli X rotation gates $R_l(\gamma_l) = \prod_{j=1}^n R_{X,j}(\gamma_{l,j})$ for $l=1,\ldots, L+1$. Let the final set of rotations be Pauli Y rotations $R_{L+2}(\gamma_{L+2}) = \prod_{j=1}^n R_{Y,j}(\gamma_{L+2,j})$. Let $U_l'$ be any unitaries that commute with all single-qubit Pauli Z operators.\\

Let all single qubit rotation parameters $\gamma_{l,j}$ be drawn independently from the normal distribution $\mathcal{N}(0,\tau^2)$ and let the initial state $\rho_0$ be the zero state $\ket{0^n}\bra{0^n}$. Assume that $U_l'$ are also drawn from a (possibly trivial) probability distribution independently from $\gamma_{l,j}$.\\

In the generative circuit we propose in the main text, we have logarithmic depth $L = O(\log n)$ and constant variance $\tau^2$, and the entangling unitaries $U_l'$ consist of CZ gates between qubit pairs. We want to show that the variance of all $k$-local observables is inverse polynomially large. First, we will determine the expectation.

\begin{theorem} \label{th:expectation}
Let $P$ be a Pauli $Z$-operator acting on $k$ qubits, so $P \in \mathcal{P}_Z := \{P_1 \otimes \cdots \otimes P_n : P_j \in \{I,Z\}\}$. Then
\[
\mathbb{E}_{\gamma\sim \Gamma} [\mathrm{Tr}(PU(\gamma)\rho_0U(\gamma)^\dagger)] =
\left(e^{-\tau^2/2}\right)^{k(L+2)}.
\]
For any Pauli operator $P' \notin \mathcal{P}_Z$, we have instead
\[
\mathbb{E}_{\gamma \sim \Gamma} [\mathrm{Tr}(P'U(\gamma)\rho_0U(\gamma)^\dagger)] =
0.\\
\]
\end{theorem}

In the proof, we will use the expectations of trigonometric functions of normal variables
\begin{equation} \label{eq:trig_expect}
\mathbb{E}_{\gamma \sim \mathcal{N}(0,\tau^2)} [\cos(\gamma)] = e^{-\tau^2/2}, \qquad \mathbb{E}_{\gamma \sim \mathcal{N}(0,\tau^2)} [\sin(\gamma)] = 0,
\end{equation}
as well as the following lemma.

\begin{lemma}\label{le:preserve}
Assume that $\mathrm{Tr}(OP) = 0$ for all $P \in \mathcal{P}_Z$ and that V is a unitary that commutes with all Pauli Z-operators. Then $\mathrm{Tr}(V^{\dagger} O V P) = 0$ for all $P \in \mathcal{P}_Z$. Also for any qubit $j$, we have
\[
\mathbb{E}_{\gamma \sim \mathcal{N}(0,\tau^2)}[\mathrm{Tr}(R_{X,j}(\gamma)^\dagger O R_{X,j}(\gamma) P)] = 0.
\]
\end{lemma}

\begin{proof}
Let $P \in \mathcal{P}_Z$ be arbitrary. For the first claim, we compute directly that for any $P$, we have $\mathrm{Tr}(V^{\dagger} O V P) = \mathrm{Tr}(O V P V^{\dagger}) = \mathrm{Tr}(OP) = 0$.\\

For the second claim, we have two cases. If $P_j = I$, then $P$ commutes with $R_{X,j}(\gamma)$, and we obtain immediately that 
\[
\mathrm{Tr}(R_{X,j}(\gamma)^\dagger O R_{X,j}(\gamma) P) = \mathrm{Tr}(O R_{X,j}(\gamma) P R_{X,j}(\gamma)^\dagger) = \mathrm{Tr}(OP) = 0.
\]
If instead $P_j = Z$, we have that $R_X(\gamma)_j^\dagger O R_X(\gamma)_j = \cos(\gamma)O - \sin(\gamma) iXO$. Therefore we have
\[
\mathbb{E}_{\gamma \sim \mathcal{N}(0,\tau^2)}[\mathrm{Tr}(R_X(\gamma)_j^\dagger O R_X(\gamma)_j P)] = \mathbb{E}_{\gamma \sim \mathcal{N}(0,\tau^2)}[\cos(\gamma)\mathrm{Tr}(O P)] + \mathbb{E}_{\gamma \sim \mathcal{N}(0,\tau^2)}[\sin(\gamma)\mathrm{Tr}(iXO P)] = 0.
\]

\end{proof}

\begin{proof}[Proof of Theorem \ref{th:expectation}]
We will show the result for fixed unitaries $U_l'$, because we can use the fact that $U_l'$ are independent from $\gamma_{l,j}$ to write
\[
\mathbb{E} [\mathrm{Tr}(PU(\gamma)\rho_0U(\gamma)^\dagger)] = \mathbb{E}_{U_l':l=1,\ldots,L} \mathbb{E}_{\gamma \sim \Gamma} \,\mathrm{Tr}(PU(\gamma)\rho_0U(\gamma)^\dagger)
\]
to deduce the general case.\\

Define $\rho_l$ to be $\rho_l = U_l(\gamma_l) \rho_{l-1} U_l(\gamma_l)^\dagger$ for $l=1,\ldots, L$ and $\rho_{l} = R_{l}(\gamma_{l}) \rho_{l-1} R_{l}(\gamma_{l})^\dagger$ for $l=L+1, L+2$. 
We have that
\[
\mathrm{Tr}(PU(\gamma)\rho_0U(\gamma)^\dagger) =  \mathrm{Tr}(P\rho_{L+2}) 
\]
Let $P \in \mathcal{P}_Z$ be a Pauli operator on $k$ qubits. Then both Pauli $X$ and $Y$ operators anti-commute with $P$ exactly on those $k$ qubits.
At any index $j$ not in the support of $P$, we have for any state $\rho$
\[
\mathbb{E}_{\gamma_{L+2,j}} [\mathrm{Tr}(P R_{Y,j}(\gamma_{L+2,j}) \rho R_{Y,j}(\gamma_{L+2,j})^\dagger)] = \mathbb{E}_{\gamma_{L+2,j}} [\mathrm{Tr} (P \rho)] = \mathrm{Tr} (P \rho)
\]
If $j$ is inside the support of $P$ then we have instead
\[
\mathbb{E}_{\gamma_{L+2,j}} [\mathrm{Tr}(P R_{Y,j}(\gamma_{L+2,j}) \rho R_{Y,j}(\gamma_{L+2,j})^\dagger)] = \mathbb{E}_{\gamma_{L+2,j}} [\left(\cos(\gamma_{L+2,j})  \mathrm{Tr}(P\rho) + \sin(\gamma_{L+2,j}) \mathrm{Tr}(iY_{j}P)  \right)] = \left(e^{-\tau^2/2}\right)\mathrm{Tr}(P\rho)
\]
by \ref{eq:trig_expect}. Therefore
\[
\mathbb{E}_{\gamma_{L+2}}[\mathrm{Tr}(P\rho_{L+2})] = \left(e^{-\tau^2/2}\right)^k\mathrm{Tr}(P\rho_{L+1}).
\]
An analogous argument for $X$-rotations yields that
\[
\mathbb{E}_{\gamma_{L+1}}[\mathrm{Tr}(P\rho_{L+1})] = \left(e^{-\tau^2/2}\right)^k\mathrm{Tr}(P\rho_{L}).
\]
For $l=1, \ldots, L$, we can use the fact that $U_l'$ commutes with $P$ to obtain
\begin{align*}
\mathbb{E}_{\gamma_{l}}[\mathrm{Tr}(P\rho_{l})] &= \mathbb{E}_{\gamma_{l}}[\mathrm{Tr}(PU_l'R_l(\gamma_l)\rho_{l-1}R_l(\gamma_l)^\dagger U_l'^\dagger)] = \mathbb{E}_{\gamma_{l}}[\mathrm{Tr}(U_l'^\dagger PU_l'R_l(\gamma_l)\rho_{l-1}R_l(\gamma_l)^\dagger)]\\
&= \mathbb{E}_{\gamma_{l}}[\mathrm{Tr}(PR_l(\gamma_l)\rho_{l-1}R_l(\gamma_l)^\dagger)] = \left(e^{-\tau^2/2}\right)^k\mathrm{Tr}(P\rho_{l-1}).
\end{align*}

We combine these estimates to obtain that 
\begin{align*}
\mathbb{E}_{\gamma\sim \Gamma} [\mathrm{Tr}(PU(\gamma)\rho_0U(\gamma)^\dagger)] &= \mathbb{E}_{\gamma_{1}} \cdots \mathbb{E}_{\gamma_{L}} \mathbb{E}_{\gamma_{L+1}} \mathbb{E}_{\gamma_{L+2}} \mathrm{Tr}(P\rho_{L+2}) = \mathbb{E}_{\gamma_{1}} \cdots \mathbb{E}_{\gamma_{L}} \mathbb{E}_{\gamma_{L+1}} \left(e^{-\tau^2/2}\right)^k\mathrm{Tr}(P\rho_{L+1})\\
&= \mathbb{E}_{\gamma_{1}} \cdots \mathbb{E}_{\gamma_{L}}  \left(e^{-\tau^2/2}\right)^{2k}\mathrm{Tr}(P\rho_{L}) = \cdots = \left(e^{-\tau^2/2}\right)^{k(L+2)}\mathrm{Tr}(P\rho_{0}) = \left(e^{-\tau^2/2}\right)^{k(L+2)}
\end{align*}
This proves the first claim.\\

For the second claim, note that we can write the zero state as a linear combination
\[
\rho_0 = 2^{-n} \sum_{P \in \mathcal{P}_Z} P.
\]
Let $P' \notin \mathcal{P}_Z$ be a Pauli operator. Then $\mathrm{Tr}(P'P) = 0$ for all $P \in \mathcal{P}_Z$. Using Lemma \ref{le:preserve} repeatedly for each gate of $U(\gamma)$, we can deduce that
\[
\mathbb{E}_{\gamma \sim \Gamma} [\mathrm{Tr}(U(\gamma)^\dagger P'U(\gamma)P)] = 0
\]
for all $P \in \mathcal{P}_Z$. It follows that 
\begin{align*}
\mathbb{E}_{\gamma \sim \Gamma} [\mathrm{Tr}(P'U(\gamma)\rho_0U(\gamma)^\dagger)] &= \mathbb{E}_{\gamma \sim \Gamma} [\mathrm{Tr}(U(\gamma)^\dagger P'U(\gamma)\rho_0)] = 2^{-n} \sum_{P \in \mathcal{P}_Z} \mathbb{E}_{\gamma \sim \Gamma} [\mathrm{Tr}(U(\gamma)^\dagger P'U(\gamma)P)] = 0.
\end{align*}
This proves the second claim.
\end{proof}

Next we need to estimate the expectation of the square, $\mathbb{E}_{\gamma \sim \Gamma} [\mathrm{Tr}(PU(\gamma)\rho_0U(\gamma)^\dagger)^2]$. This will allow us to lower bound the variance.

\begin{theorem} \label{th:sqr_expec}
Let $P$ be a Pauli operator that is the product of $k_1$ $X$-operators, $k_2$ $Y$-operators and $k_3$ $Z$-operators. Then we have
\[
\mathbb{E}_{\gamma \sim \Gamma} [\mathrm{Tr}(PU(\gamma)\rho_0U(\gamma)^\dagger)^2] \geq
\left[ \frac{1}{2}(1-\exp(-2\tau^2)) \right]^{k_1+k_2} \left[ \frac{1}{2}(1+\exp(-2\tau^2)) \right]^{k_2+2k_3+kL}.
\]
\end{theorem}

In the proof, we will use the following identities for expectations:
\begin{equation} \label{eq:trig_sqr_expect}
\mathbb{E}_{\gamma \sim \mathcal{N}(0,\tau^2)} [\cos^2(\gamma)] = \frac{1}{2}(1+e^{-2\tau^2}), \qquad \mathbb{E}_{\gamma \sim \mathcal{N}(0,\tau^2)} [\sin^2(\gamma)] = \frac{1}{2}(1-e^{-2\tau^2})    
\end{equation}

\begin{proof}
We may again assume that the gates $U_l'$ are fixed.
Let $P$ be a Pauli operator. As in the previous proof, define $\rho_l$ to be $\rho_l = U_l(\gamma_l) \rho_{l-1} U_l(\gamma_l)^\dagger$ for $l=1,\ldots, L$ and $\rho_{l} = R_{l}(\gamma_{l}) \rho_{l-1} R_{l}(\gamma_{l})^\dagger$ for $l=L+1, L+2$.\\

Consider the rotations $R_{Y,j}(\gamma_{L+2,j})$. We see that if $P_j \in \{I,Y\}$, then $P_j$ commutes with $R_{Y,j}(\gamma_{L+2,j})$ and therefore
\[
\mathrm{Tr}(PR_{Y,j}(\gamma_{L+2,j}) \rho R_{Y,j}(\gamma_{L+2,j})^\dagger) = \mathrm{Tr}(P \rho ).
\]
If instead $P_j \in \{X,Z\}$, then $P_j$ anti-commutes with  $Y_j$ and we obtain that
\[
\mathrm{Tr}(PR_{Y,j}(\gamma_{L+2,j}) \rho R_{Y,j}(\gamma_{L+2,j})^\dagger) = \cos(\gamma_{L+2,j}) \mathrm{Tr}(P \rho ) + \sin(\gamma_{L+2,j}) \mathrm{Tr}(i Y_j P \rho ).
\]
We can use \ref{eq:trig_sqr_expect} to estimate the expectation of the square by 
\[
\mathbb{E}_{\gamma_{L+2,j} \sim \mathcal{N}(0,\tau^2)} [\mathrm{Tr}(PR_{Y,j}(\gamma_{L+2,j}) \rho R_{Y,j}(\gamma_{L+2,j})^\dagger)^2] \geq \mathbb{E}_{\gamma_{L+2,j} \sim \mathcal{N}(0,\tau^2)} [\cos^2(\gamma_{L+2,j})\mathrm{Tr}(P \rho )^2] = \frac{1}{2}(1+e^{-2\tau^2})\mathrm{Tr}(P \rho )^2
\]
when $P_j = Z$ and 
\begin{align*}
\mathbb{E}_{\gamma_{L+2,j} \sim \mathcal{N}(0,\tau^2)} [\mathrm{Tr}(PR_{Y,j}(\gamma_{L+2,j}) \rho R_{Y,j}(\gamma_{L+2,j})^\dagger)^2] &\geq \mathbb{E}_{\gamma_{L+2,j} \sim \mathcal{N}(0,\tau^2)} [\sin^2(\gamma_{L+2,j})\mathrm{Tr}(i Y_j P \rho )^2]\\ &= \frac{1}{2}(1-e^{-2\tau^2})\mathrm{Tr}(i Y_j P \rho )^2
\end{align*}
when $P_j = X$. This yields that
\[
\mathbb{E}_{\gamma_{L+2}} [\mathrm{Tr}(P \rho_{L+2})^2] \geq \left[ \frac{1}{2}(1-\exp(-2\tau^2)) \right]^{k_1} \left[ \frac{1}{2}(1+\exp(-2\tau^2)) \right]^{k_3}  \mathrm{Tr}( \prod_{j: P_j = X}(iY_j) P \rho_{L+1})^2.
\]
Note that the operator $\prod_{j: P_j = X}(iY_j) P$ inside the trace is a Pauli operator obtained from $P$ by replacing $X$ with $Z$ at all $k_1$ indices. This Pauli operator consists of $k_1+k_3$ $Z$-operators and $k_2$ $Y$-operators. Doing similar arguments for $R_{X,j}$ gates, we can estimate that
\[
\mathbb{E}_{\gamma_{L+1}} [\mathrm{Tr}(P \rho_{L+1})^2] \geq \left[ \frac{1}{2}(1-\exp(-2\tau^2)) \right]^{k_2} \left[ \frac{1}{2}(1+\exp(-2\tau^2)) \right]^{k_1+k_3}  \mathrm{Tr}( P' \rho_{L})^2.
\]
Here $P' \in \mathcal{P}_Z$ is the Pauli operator obtained from $P$ by replacing every $X$ and $Y$ operator from $P$ with a $Z$ operator, acting on $k$ qubits in total.\\

As $U_l'$ commutes with $P'$, we can estimate our quantity for $1 \leq l \leq L$ repeatedly using the $P_j=Z$ bound approach as above:
\begin{align*}
\mathbb{E}_{\gamma_{l}} [\mathrm{Tr}( P' \rho_{l})^2] &= \mathbb{E}_{\gamma_{l}} [\mathrm{Tr}( P' U_l' R_{l}(\gamma_l) \rho_{l-1} R_l(\gamma_l)^\dagger U_l'^\dagger )^2] = \mathbb{E}_{\gamma_{l}} [\mathrm{Tr}( P' R_{l}(\gamma_l) \rho_{l-1} R_l(\gamma_l)^\dagger )^2]\\
&\geq \left[ \frac{1}{2}(1+\exp(-2\tau^2)) \right]^k \mathrm{Tr}( P'  \rho_{l-1} )^2.
\end{align*}
Combining these estimates yields that 
\begin{align*}
\mathbb{E}_{\gamma_{1}} \mathbb{E}_{\gamma_{2}} \cdots  \mathbb{E}_{\gamma_{L+2}} \mathrm{Tr}(P \rho_{L+2})^2 &\geq \mathbb{E}_{\gamma_{1}} \mathbb{E}_{\gamma_{2}} \cdots  \mathbb{E}_{\gamma_{L+1}} \left[ \frac{1}{2}(1-\exp(-2\tau^2)) \right]^{k_1} \left[ \frac{1}{2}(1+\exp(-2\tau^2)) \right]^{k_3}  \mathrm{Tr}( \prod_{j: P_j = X}(iY_j) P \rho_{L+1})^2\\
& \geq \mathbb{E}_{\gamma_{1}} \mathbb{E}_{\gamma_{2}} \cdots  \mathbb{E}_{\gamma_{L}}\left[ \frac{1}{2}(1-\exp(-2\tau^2)) \right]^{k_1 + k_2} \left[ \frac{1}{2}(1+\exp(-2\tau^2)) \right]^{k_1+2k_3}  \mathrm{Tr}( P' \rho_{L})^2 \\
& \geq \mathbb{E}_{\gamma_{1}} \mathbb{E}_{\gamma_{2}} \cdots \mathbb{E}_{\gamma_{L-1}}\left[ \frac{1}{2}(1-\exp(-2\tau^2)) \right]^{k_1 + k_2} \left[ \frac{1}{2}(1+\exp(-2\tau^2)) \right]^{k_1+2k_3 + k}  \mathrm{Tr}( P' \rho_{L-1})^2\\
& \geq \left[ \frac{1}{2}(1-\exp(-2\tau^2)) \right]^{k_1 + k_2} \left[ \frac{1}{2}(1+\exp(-2\tau^2)) \right]^{k_1+2k_3 + kL}  \mathrm{Tr}( P' \rho_{0})^2\\
 &= \left[ \frac{1}{2}(1-\exp(-2\tau^2)) \right]^{k_1 + k_2} \left[ \frac{1}{2}(1+\exp(-2\tau^2)) \right]^{k_1+2k_3 + kL}.
\end{align*}
This proves the claim.
\end{proof}

As variance can be expressed $\mathrm{Var}[X]=\mathbb{E}[X^2]-\mathbb{E}[X]^2$, since we proved in Theorem \ref{th:expectation} that the expectation of all Pauli observables with at least one $X$- or $Y$-operator vanishes, for these observables our lower bound on $\mathbb{E}[X^2]$ becomes directly a bound on the variance, proving that their variance is inverse polynomially large when $\tau^2 = O(\log n / L)$. We now show the same for the products of $Z$-operators.

\begin{theorem} \label{th:z_variance}
For a Pauli operator $P \in \mathcal{P}_Z$ with $k$ non-identity terms, we have
\[
\mathrm{Var}_{\gamma \sim \Gamma} [\mathrm{Tr}(P U(\gamma) \rho_0 U_\gamma^\dagger] \geq \frac{k(L+2)}{2} (1- \exp(-\tau^2))^2\exp(-\tau^2)^{k(L+2)-1}.
\]
\end{theorem}

\begin{proof}
Combining Theorems \ref{th:expectation} and \ref{th:sqr_expec} yields that
\begin{align*}
\mathrm{Var}_{\gamma \sim \Gamma} [\mathrm{Tr}(P U(\gamma) \rho_0 U(\gamma)^\dagger] &= \mathbb{E}_{\gamma \sim \Gamma} \left[\mathrm{Tr}(P U(\gamma) \rho_0 U(\gamma)^\dagger)^2\right] - \left[\mathbb{E}_{\gamma \sim \Gamma} [\mathrm{Tr}(P U(\gamma) \rho_0 U(\gamma)^\dagger)]\right]^2\\
&\geq \left[ \frac{1}{2}(1+\exp(-2\tau^2)) \right]^{k(L+2)} - \exp({-\tau^2})^{k(L+2)}.
\end{align*}
For $a \geq b \geq 0$, we can estimate $a^N - b^N = (a-b)(a^{N-1} + a^{N-2}b + \cdots + ab^{N-2} + b^{N-1}) \geq (a-b)\cdot Nb^{N-1}$. Therefore
\[
\left[ \frac{1}{2}(1+\exp(-2\tau^2)) \right]^{k(L+2)} - \exp({-\tau^2})^{k(L+2)} \geq \frac{k(L+2)}{2}\left[ 1+\exp(-2\tau^2) - 2\exp({-\tau^2}) \right] \exp(-\tau^2)^{k(L+2)-1}.
\]
Noting that $1+\exp(-2\tau^2) - 2\exp({-\tau^2}) = (1-\exp(-\tau^2))^2$ finishes the proof.
\end{proof}

Combining Theorems \ref{th:sqr_expec} and \ref{th:z_variance} immediately yields the following corollary.

\begin{corollary} \label{cor:variance}
If $\tau^2 = O(\log n / L )$, then for any $k$-local Pauli observable we have $\mathrm{Var}_ {\gamma \sim \Gamma} [\mathrm{Tr}(P U(\gamma) \rho_0 U(\gamma)^\dagger)] \geq 1/\mathrm{poly}(n)^k$.
\end{corollary}
In particular, having $L = O(\log n)$ and constant $\tau^2$ implies Theorem \ref{th:main} from the main text. For single qubit $Z$ observables, the expectation is $(e^{-\tau^2/2})^{L+2}$ by Theorem \ref{th:expectation}, meaning that we can use scaling factor $e^{-L\tau^2/2}$ when generating Figure \ref{fig:pauli_prop_comb}.\\

Our proof relies on the assumption that the entangling unitaries are diagonal in the $Z$ basis. In \cite{changLatentStylebasedQuantum2024}, the setup is similar to ours but the assumption is instead that the entangling gates are Clifford gates. As the model we are proposing has CZ gates in the entangling layer, either of the two proofs would suffice. Our proof holds not only for CZ gates, but for arbitrary unitaries that are diagonal in $Z$ basis, even if they are not Clifford gates. However, if the entangling gates are Clifford gates that do not commute with $Z$-operators, our proof no longer works and instead the proof from \cite{changLatentStylebasedQuantum2024} is more appropriate.


\section{Subvolume Law versus Area Law}\label{app:barren}
In \cite{Leone_HEA_2024}, the authors explain how, for a variational circuit ansatz, the initial state fed into the ansatz may induce barren plateaus even if the model would not suffer from them otherwise. In the case of the shallow hardware efficient ansatz (HEA), the authors call the requirement the initial state has to satisfy to prevent a barren plateau an {\it area law}. In this appendix, we argue that this should instead be called a {\it (weak) subvolume law}, to avoid conflation with the area law from many-body literature.\\

The standard definition of the area law for the entanglement from quantum many-body physics requires that for a subsystem $\Lambda$ the entropy of entanglement, i.e. the von Neumann entropy of the reduced density matrix, scales with the size of the boundary, i.e.
\begin{equation}
S(\rho_{\Lambda}) := -\mathrm{Tr} [ \rho_\Lambda \log \rho_{\Lambda} ] \sim |\partial \Lambda |
\end{equation}
For a volume law, one instead has that the entropy scales with the size, $S(\rho_{\Lambda}) \sim |\Lambda|$. In the barren plateau literature, the relevant quantity to consider is the distance of the reduced density matrix from the maximally mixed state. As noted by the authors of \cite{Leone_HEA_2024}, this is non-standard and leads to different definitions: an area law requires $|\Lambda|-S(\rho_{\Lambda}) \geq 1/\mathrm{poly}(n)$ and a volume law $|\Lambda|-S(\rho_{\Lambda}) \leq 1/\exp(n)$.\\

To avoid confusion related to these terms, we shall adopt the following convention. We say that a state $\rho$ possesses \emph{subvolume law} entanglement within $\Lambda$ if
\begin{equation}
S(\rho_{\Lambda}) \in o(|\Lambda|),
\end{equation}
and it possesses \emph{weak subvolume law} if
\begin{equation}
S(\rho_{\Lambda}) \leq |\Lambda|-\frac{1}{P(n)}
\end{equation}
for a polynomial $P$. The definition for the weak subvolume law coincides with the definition of area law in barren plateau literature. Note that the subvolume law implies weak subvolume law.\\

An information theoretical measure of pure state entanglement is distinguishability of a subsystem from the maximally mixed state
\begin{equation}
\mathcal{I}_{\Lambda}(\rho) := \|\rho_{\Lambda} - \frac{1}{2^{|\Lambda|}}\mathbb{I}\|_1,
\end{equation}
where the norm is taken to be the Schatten 1-norm, the sum of singular values. For any observable $O$ acting non-trivially on $\Lambda$, we have that $|\text{tr}(O\rho) - 2^{-n}\mathrm{Tr}[O]| \leq \mathcal{I}_{\Lambda}(\rho) \|O\|_\infty$.\\

Using the continuity of entropy and the fact that the maximum at maximally mixed state is a critical point, we see that when $|\Lambda| \in O(\log n)$ the state possesses weak subvolume law entanglement within $\Lambda$ if and only if $\mathcal{I}_{\Lambda}(\rho) \in \Omega(1/\mathrm{poly}(n))$. This is done by Taylor expansion around the critical point, where we see that $|\Lambda|-S(\rho_{\Lambda}) \sim \|\rho_{\Lambda} - \frac{1}{2^{|\Lambda|}}\mathbb{I}\|_2^2$ in Hilbert-Schmidt distance. For logarithmic $|\Lambda|$ the Schatten 1-norm is equivalent to the 2-norm up to a polynomial factor.\\

In order for a circuit to be trainable, we want that the gradient of the cost function is large, meaning $|\nabla f| \geq 1/\mathrm{poly}(n)$, for a large portion of the parameter space. To ensure this, one needs a lower bound on partial derivatives.\\

Theorem 2 in \cite{cerezo2021cost} gives a lower bound for the variance of partial derivatives for the cost function $f=\sum_{i}c_i\mathrm{Tr}[ U^{\dagger} O_iU\rho]$ as
\begin{equation}
\mathrm{Var} (\partial_\nu f) \geq c^{D} \sum_{i \in i_{\mathcal{L}}} \sum_{(k,k') \in k_{\mathcal{L}_B}} c_i^2 \epsilon(\rho_{(k,k')}) \epsilon(O_i),
\end{equation}
where $O_i$ are acting on qubits in the forward light cone, the subsystems $(k,k')$ are contained in the backward light cone, and $\epsilon(A) := \|A-\mathrm{tr}(A)\mathbb{I}/\dim (A)\|_2$ is Hilbert-Schmidt distance from the normalized identity.\\

Here we see the role of entanglement: if one of the $\epsilon(\rho_{(k,k')})$ is polynomially large, then the variance of the partial derivative is polynomially large and we have no barren plateau. Using the relation between norms we obtain the bound $\epsilon(\rho_{(k,k')}) \geq 2^{-|(k,k')|} \mathcal{I}_{(k,k')}(\rho)$, and we see that the weak subvolume law on logarithmic subsystems of the initial state implies absence of barren plateaus for HEA - from our assumptions that the circuit is of logarithmic depth and the final observable is local, the prefactor scales inverse polynomially with $n$.\\

\subsection{Preventing Barren Plateaus in our Model}

Our generative circuit has local observables of expected size $\mathrm{poly}(n)^{-1}$, which directly implies a polynomial lower bound on distinguishability from the maximally mixed state on local subsystems, by using the variable in question. The states produced by our generative circuit are therefore suitable initial states for a shallow HEA in order to avoid barren plateaus, since this distinguishability is exactly the required criterion.\\

In general, any circuit that avoids mode collapse for local observables also produces states that are on average locally distinguishable from the maximally mixed state, again implying suitability for the HEA ansatz.

\begin{theorem}
Let $\rho(\gamma)$ be a state with $\mathrm{Var}_{\gamma \sim \Gamma}[\mathrm{Tr}(O\rho(\gamma))] \geq \mathrm{poly}(n)^{-1}$ for $k$-local observables $O$. Then \[\mathbb{E}_{\gamma \sim \Gamma} [\mathcal{I}_{(k_1,k_2)}(\rho(\gamma))^2] \geq \mathrm{poly}(n)^{-1}.\]
\end{theorem}

\begin{proof}
Using the duality of Schatten 1-norm, we see that 
\[
\mathcal{I}_{(k_1,k_2)}(\rho(\gamma)) = \sup_{\|A\|_{\infty} = 1} \mathrm{Tr}\left[\left(\rho(\gamma)_{(k_1,k_2)}- \frac{\mathbb{I}_{(k_1,k_2)}}{2^{|k_2-k_1+1|}}\right)A^\dagger\right]
\]
Let $O = O_1 \otimes \mathbb{I}$ be any $k$-local Pauli observable supported on $(k_1,k_2)$. Then $\|O_1\|_\infty = 1$ and 
\[
\sup_{\|A\|_{\infty} = 1} \mathrm{Tr}\left[\left(\rho(\gamma)_{(k_1,k_2)}- \frac{\mathbb{I}_{(k_1,k_2)}}{2^{|k_2-k_1+1|}}\right)A^\dagger\right] \geq \left|\mathrm{Tr}\left[\left(\rho(\gamma)_{(k_1,k_2)}- \frac{1}{2^{|k_2-k_1+1|}}\mathbb{I}_{(k_1,k_2)}\right)O_1^\dagger\right]\right| = |\mathrm{Tr}[\rho(\gamma)O]|.
\]
Therefore
\[
\mathbb{E}_{\gamma\sim\Gamma} \left[\mathcal{I}_{(k_1,k_2)}(\rho(\gamma))^2\right] \geq \mathbb{E}_{\gamma \sim \Gamma} \left[\mathrm{Tr}(\rho(\gamma)O)^2\right] \geq \mathrm{Var}_{\gamma\sim \Gamma}[\mathrm{Tr}(O\rho(\gamma))].
\]
\end{proof}

\section{Truncation Strategies for Pauli Propagation}\label{app:Trunc_Methods}

In this appendix we go through the four truncation methods for Pauli propagation that we tested in Section \ref{subsec:pauli_prop}. Each of these strategies aims to approximate the Pauli sum propagated through the circuit by truncating terms, reducing the computational complexity at the expense of accuracy of the algorithm. All four of these methods are implemented in the Julia library presented in \cite{rudolph2025pauli}.\\

\emph{Small coefficient truncation.} This simple truncation strategy truncates terms in the Pauli sum once the absolute value of their coefficients falls below the specified threshold. The key weakness of this approach is that the number of dropped terms may be very large, causing significant contribution to the final error.\\

\emph{Pauli weight truncation.} Under this strategy we truncate the terms of the Pauli sum where the Pauli string has many non-identity Paulis, the idea being that the overlap is likely to be small for such Pauli strings. This method is effective on random inputs invariant under single qubits on average \cite{angrisani2024classically}. This does not cover the small-angle case, however. In our generative circuit, the structure of the CZ gates in Erdős–Rényi graph ensures that the Pauli sum will have terms with logarithmic number of non-identity Paulis. In particular, gates in the initial layer are likely to create high weight contributions from polynomially significant low weight coefficients, and each subsequent layer can do this also. This means that we would need to increase the number of non-identity Paulis we allow. The number of stored Pauli terms with weight up to $W$ can scale as $O(n^W)$.\\

\emph{Sine truncation.} In this truncation strategy, supposing the circuit is composed of Clifford gates and Pauli rotations, we keep track of the number of sine factors a Pauli path picks up when branching at Pauli rotations, truncating the terms that pick up too many sine factors. This strategy works well when the angles in the circuit are small \cite{Lerch2024}. The number of terms that need to be kept if $S$ sine factors are allowed grows like $O(n^S)$.\\

\emph{Frequency truncation.} As with sine truncation, we again track the branching of a Pauli path, but with this strategy we instead truncate the paths with high frequency: the total number of cosine and sine factors, which is equal to the number of branchings at Pauli rotations. This can efficiently estimate noisy variational circuits \cite{fontana2025classical}. The number of Pauli terms that the algorithm keeps track of when maximum allowed frequency is $F$ is bounded by $O(2^F)$, so this strategy is polynomial for $F = O(\log n)$.\\

These four are not the only possible truncation strategies. For example, path-weight truncation \cite{aharonov2023noisy} is not considered in this paper due to being best known for effectiveness in simulating noisy quantum circuits.



\end{document}